\documentclass[10pt, conference]{IEEEtran}
\usepackage{amsmath}
\usepackage{amssymb}
\usepackage{amsfonts}
\usepackage{amscd}
\usepackage{mathrsfs}
\usepackage[final]{graphicx}
\usepackage{graphicx}
\usepackage{color}
\usepackage{url}

\newtheorem{lemma}{Lemma}
\newtheorem{definition}{Definition}
\newtheorem{proposition}{Proposition}

\newtheorem{example}{Example}

\begin{document}

\title{Construction of Block Orthogonal STBCs and Reducing Their Sphere Decoding Complexity}

\author{
\authorblockN{G. R. Jithamithra and B. Sundar Rajan,\\}
\authorblockA{Dept. of ECE, Indian Institute of Science, \\
Bangalore 560012, India\\
Email:\{jithamithra,bsrajan\}@ece.iisc.ernet.in\\
}
}
\maketitle


\begin{abstract}
Construction of high rate Space Time Block Codes (STBCs) with low decoding complexity has been studied widely using techniques such as sphere decoding and non Maximum-Likelihood (ML) decoders such as the QR decomposition decoder with M paths (QRDM decoder). Recently Ren et al., presented a new class of STBCs known as the block orthogonal STBCs (BOSTBCs), which could be exploited by the QRDM decoders to achieve significant decoding complexity reduction without performance loss. The block orthogonal property of the codes constructed was however only shown via simulations. In this paper, we give analytical proofs for the block orthogonal structure of various existing codes in literature including the codes constructed in the paper by Ren et al. We show that codes formed as the sum of Clifford Unitary Weight Designs (CUWDs) or Coordinate Interleaved Orthogonal Designs (CIODs) exhibit block orthogonal structure. We also provide new construction of block orthogonal codes from Cyclic Division Algebras (CDAs) and Crossed-Product Algebras (CPAs). In addition, we show how the block orthogonal property of the STBCs can be exploited to reduce the decoding complexity of a sphere decoder using a depth first search approach. Simulation results of the decoding complexity show a 30\% reduction in the number of floating point operations (FLOPS) of BOSTBCs as compared to STBCs without the block orthogonal structure. 
\end{abstract}


\section{Introduction \& Preliminaries}
\label{sec1}
Consider a minimal-delay space-time coded Rayleigh quasi-static flat fading MIMO channel with full channel state information at the receiver (CSIR). The input output relation for such a system is given by
\begin{equation}
\label{system_model}
\textbf{Y} = \textbf{H}\textbf{X} + \textbf{N},
\end{equation}
where $\textbf{H} \in \mathbb{C}^{n_{r} \times n_{t}}$ is the channel matrix and $\textbf{N} \in \mathbb{C}^{n_{r} \times n_{t}}$ is the additive noise. Both $\textbf{H}$ and $\textbf{N}$ have entries that are i.i.d. complex-Gaussian with zero mean and variance 1 and $N_{0}$ respectively. The transmitted codeword is $\textbf{X} \in \mathbb{C}^{n_{t} \times n_{t}}$ and $\textbf{Y} \in \mathbb{C}^{n_{r} \times n_{t}}$ is the received matrix. The ML decoding metric to minimize over all possible values of the codeword $\textbf{X},$ is
\begin{equation}
\label{ML}
\textbf{M}\left( \textbf{X}\right) = \parallel \textbf{Y} - \textbf{H}\textbf{X}\parallel^{2}.
\end{equation}

\begin{definition}
\label{ld_stbc_def}
\cite{HaH}: A linear STBC $\mathcal{C}$ over a real (1-dimensional) signal set $\mathcal{S}$, is a finite set of $n_{t} \times n_t$ matrices, where any codeword matrix belonging to the code $\mathcal{C}$ is obtained  from,
\begin{equation}
\label{ld_stbc}
\textbf{X}\left( x_{1}, x_{2}, . . . , x_{K} \right) ~=~ \sum_{i = 1}^{K} x_{i}\textbf{A}_{i}, 
\end{equation}
by letting the real variables $x_{1}, x_{2}, \cdots, x_{K}$ take values from a real signal set $\mathcal{S},$ where $\textbf{A}_{i}$ are fixed $n_{t} \times n_t$ complex matrices defining the code, known as the weight matrices. The rate of this code is $\frac{K}{2n_t}$ complex symbols per channel use.
\end{definition}

We are interested in linear STBCs, since they admit sphere decoding (SD) \cite{ViB} and other QR decomposition based decoding techniques such as the QRDM decoder \cite{RGYZ} which are fast ways of decoding for the variables.
  
Designing STBCs with low decoding complexity has been studied widely in the literature. Orthogonal designs with single symbol decodability were proposed in \cite{TJC}, \cite{Li}, \cite{TiH}. For STBCs with more than two transmit antennas, these came at a cost of reduced transmission rates. To increase the rate at the cost of higher decoding complexity, multi-group decodable STBCs were introduced in \cite{DYT}, \cite{KaR}, \cite{KaR1}. Another set of low decoding complexity codes known as the fast decodable codes were studied in \cite{BHV}. Fast decodable codes have reduced SD complexity owing to the fact that a few of the variables can be decoded as single symbols or in groups if we condition them with respect to the other variables. Fast decodable codes for asymmetric systems using division algebras have been reported \cite{VHO}. The properties of fast decodable codes and multi-group decodable codes were combined and a new class of codes called fast group decodable codes were studied in \cite{RGYS}.

A new code property called the \textit{block-orthogonal} property was studied in \cite{RGYZ} which can be exploited by the QR-decomposition based decoders to achieve significant decoding complexity reduction without performance loss. This property was exploited in \cite{SiB} to reduce to the average ML decoding complexity of the Golden code \cite{BRV} and also in \cite{KaC} to reduce the worst-case complexity of the Golden code with a small performance loss. While the other low decoding complexity STBCs use the zero entries in the upper left portion of the upper triangular matrix after the QR decomposition, these decoders utilize the zeroes in the lower right portion to reduce the complexity further. 

The contributions of this paper are as follows: 
\begin{itemize}
\item We generalize the set of sufficient conditions for an STBC to be block orthogonal provided in \cite{RGYZ} for sub-block sizes greater than 1. 
\item We provide analytical proofs that the codes obtained from the sum of Clifford Unitary Weight Designs (CUWDs) \cite{RaR} exhibit the block orthogonal property when we choose the right ordering and the right number of matrices.
\item We provide new methods of construction of BOSTBCs using Coordinate Interleaved Orthogonal Designs (CIODs) \cite{KhR}, Cyclic Division Algebras (CDAs) \cite{SeRS} and Crossed Product Algebras (CPAs) \cite{ShRS} along with the analytical proofs of their block orthogonality.
\item We show that the ordering of variables of the STBC used for the QR decomposition dictates the block orthogonal structure and its parameters.  
\item We show how the block orthogonal property of the STBCs can be exploited to reduce the decoding complexity of a sphere decoder which uses a depth first search approach.
\item We provide bounds on the maximum possible reduction in the Euclidean metrics (EM) calculation during sphere decoding of BOSTBCs.
\item Simulation results show that we can reduce the decoding complexity of existing STBCs by upto 30\% by utilizing the block orthogonal property.  
\end{itemize}

The remaining part of the paper is organized as follows: In Section \ref{sec2} the system model and some known classes of low decoding complexity codes are reviewed. In Section \ref{sec3}, we derive a set of sufficient conditions for an STBC to be block orthogonal and also the effect of ordering of matrices on it. In Section \ref{sec4}, we present proofs of block orthogonal structure of various existing codes and also discuss some new methods of constructions of the same. In Section \ref{sec5}, we discuss a method to reduce the number of EM calculations while decoding a BOSTBC using a depth first search based sphere decoder and also derive bounds for the same. Simulation results for the decoding complexity of various BOSTBCs are presented in Section \ref{sec6}. Concluding remarks constitute Section \ref{sec7}. 

\indent \textit{Notations:} Throughout the paper, bold lower-case letters are used to denote vectors and bold upper-case letters to denote matrices. For a complex variable $x$, denote the real and imaginary part of $x$ by $x_{I}$ and $x_{Q}$ respectively. The sets of all integers, all real and complex numbers are denoted by $\mathbb{Z}, \mathbb{R}$ and $\mathbb{C}$, respectively. The operation of stacking the columns of $\textbf{X}$ one below the other is denoted by $vec\left(\textbf{X}\right)$. The Kronecker product is denoted by $\otimes$, $\textbf{I}_{T}$ and $\textbf{O}_{T}$ denote the $T \times T$ identity matrix and the null matrix, respectively. For a complex variable $x$, the $\check{\left(\centerdot\right)}$ operator acting on $x$ is defined as follows
\begin{equation*}
\check{x} \triangleq \left[\begin{array}{rr}
x_{I} & -x_{Q}\\
x_{Q} & x_{I}\\
\end{array}\right].
\end{equation*}
The $\check{\left(\centerdot\right)}$ operator can similarly be applied to any matrix $\textbf{X} \in \mathbb{C}^{n \times m}$ by replacing each entry $x_{ij}$ by $\check{x}_{ij}$, $i = 1, 2, \cdots, n$, $j = 1, 2, \cdots, m$, resulting in a matrix denoted by $\check{\textbf{X}} \in \mathbb{R}^{2n \times 2m}$. Given a complex vector $\textbf{x} = \left[x_{1}, x_{2}, \cdots, x_{n}\right]^{T}$, $\tilde{\textbf{x}}$ is defined as
$\tilde{\textbf{x}} \triangleq \left[x_{1I},x_{1Q},\cdots,x_{nI},x_{nQ}\right]^{T}$.


\section{System Model}
\label{sec2}

For any Linear STBC with variables $x_{1}, x_{2} . . . , x_{K}$ given by (\ref{ld_stbc}), the generator matrix $\textbf{G}$ \cite{BHV} is defined by $\widetilde{vec\left(\textbf{X}\right)} = \textbf{G} \tilde{\textbf{x}},$ 
where $\tilde{\textbf{x}} = \left[x_{1}, x_{2} . . . , x_{K}\right]^{T}$. In terms of the weight matrices, the generator matrix can be written as
\begin{equation*}
\textbf{G} = \left[\widetilde{vec\left(\textbf{A}_{1}\right)} ~ \widetilde{vec\left(\textbf{A}_{2}\right)} ~ \cdots ~ \widetilde{vec\left(\textbf{A}_{K}\right)} ~ \right].
\end{equation*}
Hence, for any STBC,  \eqref{system_model} can be written as
\begin{equation*}
\widetilde{vec\left(\textbf{Y}\right)} = \textbf{H}_{eq}\tilde{\textbf{x}} + \widetilde{vec\left(\textbf{N}\right)},
\end{equation*}
where $\textbf{H}_{eq} \in \mathbb{R}^{2n_{r}n_{t} \times K}$ is given by
$\textbf{H}_{eq} = \left(\textbf{I}_{n_{t}} \otimes \check{\textbf{H}}\right) \textbf{G},$ 
 and 
$\tilde{\textbf{x}} = \left[x_{1}, x_{2} . . . , x_{K}\right],$ 
with each $x_{i}$ drawn from a 1-dimensional (PAM) constellation. Using the above equivalent system model, the ML decoding metric \eqref{ML} can be written as
\begin{equation*}
\textbf{M}\left(\tilde{\textbf{x}}\right) = \parallel \widetilde{vec\left(\textbf{Y}\right)} - \textbf{H}_{eq}\tilde{\textbf{x}}\parallel^{2}.
\end{equation*}
Using $\textbf{Q}\textbf{R}$ decomposition of $\textbf{H}_{eq}$, we get $\textbf{H}_{eq} = \textbf{Q}\textbf{R}$ where $\textbf{Q} \in \mathbb{R}^{2n_{r}n_{t} \times K}$ is an orthonormal matrix and $ \textbf{R} \in \mathbb{R}^{K \times K}$ is an upper triangular matrix. Using this, the ML decoding metric now changes to 
\begin{equation}
\label{eq_ml_decoding_metric}
\textbf{M}\left(\tilde{\textbf{x}}\right) = \parallel \textbf{Q}^{T}\widetilde{vec\left(\textbf{Y}\right)} - \textbf{R}\tilde{\textbf{x}}\parallel^{2} = \parallel \textbf{y}^{'} - \textbf{R}\tilde{\textbf{x}}\parallel^{2}.
\end{equation}
If we have $\textbf{H}_{eq} = \left[ \textbf{h}_{1} \textbf{h}_{2} ..., \textbf{h}_{K}\right] ,$ where $ \textbf{h}_{i}, i \in 1, 2, ... , K$ are column vectors, then the $ \textbf{Q}$ and $ \textbf{R}$ matrices have the following form obtained by the Gram-Schmidt orthogonalization: 
\begin{equation}
\label{q_mat}
\textbf{Q} = \left[ \textbf{q}_{1} ~ \textbf{q}_{2} ~ ... ~ \textbf{q}_{K}\right] ,
\end{equation}
where $ \textbf{q}_{i}, i \in 1, 2, ... , K$ are column vectors, and
\begin{equation}
\label{r_mat_def}
\textbf{R} = \left[\begin{array}{ccccc}
\parallel \textbf{r}_{1} \parallel & \left\langle \textbf{q}_{1}, \textbf{h}_{2}\right\rangle & \left\langle \textbf{q}_{1}, \textbf{h}_{3}\right\rangle & \cdots & \left\langle \textbf{q}_{1}, \textbf{h}_{K}\right\rangle\\
0 & \parallel \textbf{r}_{2} \parallel & \left\langle \textbf{q}_{2}, \textbf{h}_{3}\right\rangle & \cdots & \left\langle \textbf{q}_{2}, \textbf{h}_{K}\right\rangle\\
0 & 0 & \parallel \textbf{r}_{3} \parallel & \cdots & \left\langle \textbf{q}_{3}, \textbf{h}_{K}\right\rangle\\
\vdots & \vdots & \vdots & \ddots & \vdots\\
0 & 0 & 0 & \cdots & \parallel \textbf{r}_{K} \parallel\\
\end{array}\right],
\end{equation}
where $\textbf{r}_{1} = \textbf{h}_{1},~~~~ \textbf{q}_{1} = \frac{\textbf{r}_{1}}{\parallel \textbf{r}_{1} \parallel}$ and for $i = 2, ... K,$
\begin{equation*}
\label{r_mat_entries2}
\textbf{r}_{i} = \textbf{h}_{i} - \sum_{j=1}^{i-1} \left\langle \textbf{q}_{j}, \textbf{h}_{i}\right\rangle \textbf{q}_{j} , ~~~~\textbf{q}_{i} = \frac{\textbf{r}_{i}}{\parallel \textbf{r}_{i} \parallel}.
\end{equation*}

\subsection{Low decoding complexity codes}

A brief overview of the known low decoding complexity codes is given in this section. The codes that will be described are multi-group decodable codes, fast decodable codes and fast group decodable codes. 

In case of a multi-group decodable STBC, the variables can be partitioned into groups such that the ML decoding metric is decoupled into submetrics such that only the members of the same group need to be decoded jointly. It can be formally defined as \cite{KaR}, \cite{KhR}, \cite{RaR}:

\begin{definition}
\label{multi_group_decodability}
An STBC is said to be $g$-group decodable if there exists a partition of $\left\lbrace 1, 2, ... , K\right\rbrace $ into $g$ non-empty subsets $\Gamma_{1}, \Gamma_{2}, ... , \Gamma_{g}$ such that the following condition is satisfied:
\begin{equation*}
\label{multi_group_dec_cond}
\textbf{A}_{l}\textbf{A}_{m}^{H} + \textbf{A}_{m}\textbf{A}_{l}^{H} = \textbf{0},
\end{equation*}
whenever $l \in \Gamma_{i}$ and $m \in \Gamma_{j}$ and $i \neq j$. 
\end{definition}
If we group all the variables of the same group together in \eqref{eq_ml_decoding_metric}, then the $ \textbf{R}$ matrix for the SD \cite{ViB}, \cite{DCB} in case of multi-group decodable codes will be of the following form:
\begin{equation}
\label{multi_group_r_mat}
\textbf{R} = \left[\begin{array}{cccc}
\Delta_{1} & \textbf{0} & \cdots & \textbf{0}\\
\textbf{0} & \Delta_{2} & \cdots & \textbf{0}\\
\vdots & \vdots & \ddots & \vdots\\
\textbf{0} & \textbf{0} & \cdots & \Delta_{g}\\
\end{array}\right],
\end{equation}
where $\Delta_{i}, i = 1, 2, ..., g$ is a square upper triangular matrix. 

Now, consider the standard SD of an STBC. Suppose the $ \textbf{R}$ matrix as defined in \eqref{r_mat_def} turns out to be such that when we fix values for a set of symbols, the rest of the symbols become group decodable, then the code is said to be fast decodable. Formally, it is defined as follows:
\begin{definition}
\label{fast_decodability}
An STBC is said to be fast SD if there exists a partition of $\left\lbrace 1, 2, ... , L\right\rbrace $ where $L \leq K$ into $g$ non-empty subsets $\Gamma_{1}, \Gamma_{2}, ... , \Gamma_{g}$ such that the following condition is satisfied for all $i<j$
\begin{equation}
\label{fast_decode_eq}
\left\langle \textbf{q}_{i}, \textbf{h}_{j} \right\rangle = 0,
\end{equation}
whenever $i \in \Gamma_{p}$ and $j \in \Gamma_{q}$ and $p \neq q$ where $ \textbf{q}_{i}$ and $ \textbf{h}_{j}$ are obtained from the $\textbf{Q}\textbf{R}$ decomposition of the equivalent channel matrix $\textbf{H}_{eq} = \left[ \textbf{h}_{1} \textbf{h}_{2} ..., \textbf{h}_{K}\right] = \textbf{Q}\textbf{R}$ with $ \textbf{h}_{i}, i \in 1, 2, ... , K$ as column vectors and $\textbf{Q} = \left[ \textbf{q}_{1} ~ \textbf{q}_{2} ~ ... ~ \textbf{q}_{K}\right]$ with $ \textbf{q}_{i}, i \in 1, 2, ... , K$ as column vectors as defined in \eqref{q_mat}.   
\end{definition}

Hence, by conditioning $K - L$ variables, the code becomes $g$-group decodable. As a special case, when no conditioning is needed, i.e., $L = K$, then the code is $g$-group decodable. The $ \textbf{R}$ matrix for fast decodable codes will have the following form:
\begin{equation}
\label{fast_decodable_r_mat}
\textbf{R} = \left[\begin{array}{cc}
\Delta & \textbf{B}_{1}\\
\textbf{0} & \textbf{B}_{2}\\
\end{array}\right],
\end{equation}
where $\Delta$ is an $L \times L$ block diagonal, upper triangular matrix, $ \textbf{B}_{2}$ is a square upper triangular matrix and $ \textbf{B}_{1}$ is a rectangular matrix.

Fast group decodable codes were introduced in \cite{RGYS}. These codes combine the properties of multi-group decodable codes and the fast decodable codes. These codes allow each of the groups in the multi-group decodable codes to be fast decoded. 
The $ \textbf{R}$ matrix for a fast group decodable code will have the following form:
\begin{equation}
\label{fast_group_decodable_r_mat}
\textbf{R} = \left[\begin{array}{cccc}
\textbf{R}_{1} & \textbf{0} & \cdots & \textbf{0}\\
\textbf{0} & \textbf{R}_{2} & \cdots & \textbf{0}\\
\vdots & \vdots & \ddots & \vdots\\
\textbf{0} & \textbf{0} & \cdots & \textbf{R}_{g}\\
\end{array}\right],
\end{equation}
where each $ \textbf{R}_{i}, i = 1, 2, ..., g$ will have the following form:
\begin{equation}
\label{fast_group_decodable_ri_mat}
\textbf{R}_{i} = \left[\begin{array}{cc}
\Delta_{i} & \textbf{B}_{i_{1}}\\
\textbf{0} & \textbf{B}_{i_{2}}\\
\end{array}\right],
\end{equation}
where $\Delta_{i}$ is an $L_{i} \times L_{i}$ block diagonal, upper triangular matrix, $ \textbf{B}_{i_{2}}$ is a square upper triangular matrix and $ \textbf{B}_{i_{1}}$ is a rectangular matrix.

\section{Block Orthogonal STBCs}
\label{sec3}

Block orthogonal codes introduced in \cite{RGYZ} are a sub-class of fast decodable / fast group decodable codes. They impose an additional structure on the variables conditioned in these codes. An STBC is said to be block orthogonal if the $\textbf{R}$ matrix of the code has the following structure:
\begin{equation}
\label{block_ortho_r_mat}
\textbf{R} = \left[\begin{array}{cccc}
\textbf{R}_{1} & \textbf{B}_{12} & \cdots & \textbf{B}_{1\Gamma}\\
\textbf{0} & \textbf{R}_{2} & \cdots & \textbf{B}_{2\Gamma}\\
\vdots & \vdots & \ddots & \vdots\\
\textbf{0} & \textbf{0} & \cdots & \textbf{R}_{\Gamma}\\
\end{array}\right],
\end{equation}
where each $ \textbf{R}_{i}, i = 1, 2, ..., \Gamma$ is a block diagonal, upper triangular matrix with $k$ blocks $ \textbf{U}_{i1}, \textbf{U}_{i2}, ..., \textbf{U}_{ik} $, each of size $ \gamma \times \gamma$ and $ \textbf{B}_{ij}, i = 1, 2, ..., \Gamma, ~ j = i+1, ..., \Gamma$ are non-zero matrices.

The low decoding complexity codes described in Section \ref{sec2} utilize the zero entries in the upper triangular matrix $\textbf{R}$, in the breadth first or depth first search decoders such as the sphere decoder or the QRDM decoder to achieve decoding complexity reduction. The fast sphere decoding complexity \cite{JiR} of an STBC is governed by the zeros in the upper left block of the $\textbf{R}$ matrix and does not exploit the zeros in the lower right blocks. The zeros in the lower right block can be used to reduce the average decoding complexity of the code where the average decoding complexity refers to the average number of floating operations performed by the decoder. The zeros in the lower right block are also utilized in some non ML decoders such as the QRDM decoder \cite{RGYZ} or the modified sphere decoder \cite{KaC} to reduce the decoding complexity of the code. 

\subsection{Design criteria for Block Orthogonal STBCs}
\label{bostc_results}

The structure of block orthogonal matrix was defined in \eqref{block_ortho_r_mat}. In general, the size of block diagonal matrices, $\textbf{R}_{i}$'s, and the upper triangular blocks in these matrices can be arbitrary. Similar to \cite{RGYZ}, we consider only the case that $\textbf{R}_{i}$s have the same size, $k \times k$, and the upper triangular blocks in $\textbf{R}_{i}$s each have the same size $\gamma \times \gamma$. Hence, a block orthogonal code can be represented by the parameters $\left(\Gamma, k, \gamma\right) $:
\begin{itemize}
\item $\Gamma$: The number of matrices $\textbf{R}_{i}$ in $\textbf{R}$;
\item $k$: The number of blocks in the block diagonal matrix $\textbf{R}_{i}$ - denoted by $\textbf{U}_{ij}$, $ 1 \leq j \leq k $;
\item $\gamma$: The number of diagonal entries in the matrices $\textbf{U}_{ij}$. 
\end{itemize}

A set of sufficient conditions for an STBC to be a BOSTBC with the parameters $\left( \Gamma, k, 1\right) $ are described below: 

\subsubsection{2-Block BOSTBC}
\label{two_blk_stbc}

First a condition for the STBC to be block orthogonal with parameters $\left(2, k, 1\right) $ is given. The case for $\Gamma > 2$ will be given subsequently. 

\begin{lemma}
\label{bostc_lemma1}
\cite{RGYZ} Consider an STBC of size $T \times N_{t}$ with
weight matrices $\textbf{A}_{1} , ... , \textbf{A}_{k}$ , $\textbf{B}_{1} , ... , \textbf{B}_{k}$. Let
\begin{equation*}
\mathcal{A}_{i} = \left[\begin{array}{cc}
\textbf{A}_{i}^{R} & -\textbf{A}_{i}^{I}\\
\textbf{A}_{i}^{I} & \textbf{A}_{i}^{R}\\
\end{array}\right], ~~ 
\mathcal{B}_{i} = \left[\begin{array}{cc}
\textbf{B}_{i}^{R} & -\textbf{B}_{i}^{I}\\
\textbf{B}_{i}^{I} & \textbf{B}_{i}^{R}\\
\end{array}\right]
\end{equation*}
and $\mathcal{A}_{i} \triangleq \left[ a_{iup}\right] _{2T \times 2N_{t}}$, $\mathcal{B}_{i} \triangleq \left[ b_{iup}\right] _{2T \times 2N_{t}}$, $i = 1,...,k$, $u = 1, ... 2T$ and $p = 1, ... 2N_{t}$. This STBC has block orthogonal structure $\left(2, k, 1\right) $ if the following conditions are satisfied:
\begin{itemize}
\item $\left\lbrace \mathcal{A}_{1}, ..., \mathcal{A}_{k}, \mathcal{B}_{1}, \mathcal{B}_{k}\right\rbrace $ is of dimension $2k$.
\item $\mathcal{A}_{i}^{T}\mathcal{A}_{i} = \textbf{I}$ and $\mathcal{B}_{i}^{T}\mathcal{B}_{i} = \textbf{I}$ for $i = 1,...,k$.
\item $\mathcal{A}_{i}^{T}\mathcal{A}_{j} = -\mathcal{A}_{j}^{T}\mathcal{A}_{i}$ and $\mathcal{B}_{i}^{T}\mathcal{B}_{j} = -\mathcal{B}_{j}^{T}\mathcal{B}_{i}$ for $i,j = 1,...,k$ and $i \neq j$. 
\item $\sum_{\left( p,q,s,t\right) \in \mathbb{S}} d_{pqst} = 0 $ for $i,j = 1,...,k$ and $i \neq j$ where 
\begin{equation*}
d_{pqst} = \sum_{l=1}^{k}\left( \sum_{u=1}^{2T} b_{iup}a_{lus} . \sum_{v=1}^{2T}b_{jvq}a_{lvt}\right) 
\end{equation*}
and each element (tuple) of $\mathbb{S}$ includes four uniquely permuted scalars drawn from $\left\lbrace 1, ..., 2N_{t}\right\rbrace $.  
\end{itemize}

\end{lemma}

\subsubsection{$\Gamma$-block BOSTBC, $\Gamma > 2$}
\label{gamma_blk_stbc}

The set of conditions for an STBC to have a block orthogonal structure with parameters $\left(\Gamma, k, 1\right) $  is now given.

\begin{lemma}
\label{bostc_lemma2}
\cite{RGYZ} Let the $\textbf{R}$ matrix of an STBC with weight matrices $\left\lbrace \textbf{A}_{1} , ... , \textbf{A}_{L}\right\rbrace $ , $\left\lbrace \textbf{B}_{1} , ... , \textbf{B}_{k}\right\rbrace $ be 
\begin{equation*}
\textbf{R} = \left[\begin{array}{cc}
\textbf{R}_{1} & \textbf{E}\\
\textbf{0} & \textbf{R}_{2}\\
\end{array}\right],
\end{equation*}
where $\textbf{R}_{1}$ is a $L \times L$ block-orthogonal matrix with parameters $\left(\Gamma-1, k, 1\right)$, $\textbf{E}$ is an $L \times k$ matrix and $\textbf{R}_{2}$ is a $k \times k$ upper triangular matrix. The STBC will be a block orthogonal STBC with parameters $\left(\Gamma, k, 1\right)$ if the following conditions are satisfied: 
\begin{itemize}
\item The matrices $\left\lbrace \textbf{B}_{1} , ... , \textbf{B}_{k}\right\rbrace $ are Hurwitz-Radon orthogonal.
\item The matrix $\textbf{E}$ is para-unitary, i.e., $\textbf{E}^{H}\textbf{E} = \textbf{I}$.
\end{itemize}

\end{lemma}

The authors in \cite{RGYZ} only discuss the conditions for the block orthogonal codes with parameters $\left(\Gamma, k, 1\right)$. These conditions can be easily derived for BOSTBCs with parameters $\left(\Gamma, k, \gamma\right) $ as well. We first derive the conditions for $ \Gamma = 2$. 

\begin{lemma}
\label{bostc_lemma3}
Consider an STBC of size $n_t \times T$ with weight matrices $\left\lbrace \textbf{A}_{1}, \textbf{A}_{2}, ..., \textbf{A}_{l}\right\rbrace $, $ \left\lbrace \textbf{B}_{1}, \textbf{B}_{2}, ..., \textbf{B}_{l}\right\rbrace $. Let the $ \textbf{R}$ matrix for this STBC be of the form 
\begin{equation*}
\textbf{R} = \left[\begin{array}{cc}
\textbf{R}_{1} & \textbf{E}\\
\textbf{0} & \textbf{R}_{2}\\
\end{array}\right],
\end{equation*}
where $\textbf{R}_{1}$ and $\textbf{R}_{2} $ are $l \times l$ upper triangular matrices, $\textbf{E}$ is an $l \times l$ matrix. The STBC will have a block orthogonal structure with parameters $\left(2, k, \gamma\right)$ if the following conditions are satisfied:
\begin{itemize}
\item The matrices $\left\lbrace \textbf{A}_{1} , ... , \textbf{A}_{l}\right\rbrace $ are $k$-group decodable with $\gamma$ variables in each group, i.e., $\left\lbrace \textbf{A}_{1} , ... , \textbf{A}_{l}\right\rbrace $ can be partitioned into $k$ sets $ \left\lbrace \textbf{S}_{1}, ..., \textbf{S}_{k}\right\rbrace $, each of cardinality $\gamma$ such that $ \textbf{A}_{i} \textbf{A}_{j}^{H} + \textbf{A}_{j} \textbf{A}_{i}^{H} = \textbf{0}$ for all $ \textbf{A}_{i} \in \textbf{S}_{m}$, $ \textbf{A}_{j} \in \textbf{S}_{n}$, $m \neq n$.   
\item The matrices $\left\lbrace \textbf{B}_{1} , ... , \textbf{B}_{l}\right\rbrace $ are $k$-group decodable with $\gamma$ variables in each group, i.e., $\left\lbrace \textbf{B}_{1} , ... , \textbf{B}_{l}\right\rbrace $ can be partitioned into $k$ sets $ \left\lbrace \textbf{S}_{1}, ..., \textbf{S}_{k}\right\rbrace $, each of cardinality $\gamma$ such that $ \textbf{B}_{i} \textbf{B}_{j}^{H} + \textbf{B}_{j} \textbf{B}_{i}^{H} = \textbf{0}$ for all $ \textbf{B}_{i} \in \textbf{S}_{m}$, $ \textbf{B}_{j} \in \textbf{S}_{n}$, $m \neq n$.  
\item The set of matrices $ \left\lbrace \textbf{A}_{1},..., \textbf{A}_{l}, \textbf{B}_{1},..., \textbf{B}_{l}\right\rbrace $ are such that the $ \textbf{R}$ matrix obtained has full rank.  
\item The matrix $\textbf{E}^{H}\textbf{E}$ is a block diagonal matrix with $k$ blocks of size $\gamma \times \gamma$.
\end{itemize}
\end{lemma}
\begin{proof}
Proof is given in Appendix \ref{proof_bostc_lemma3}.
\end{proof}

\begin{lemma}
\label{bostc_lemma4}
Let the $\textbf{R}$ matrix of an STBC with weight matrices $\left\lbrace \textbf{A}_{1} , ... , \textbf{A}_{L}\right\rbrace $ , $\left\lbrace \textbf{B}_{1} , ... , \textbf{B}_{l}\right\rbrace $ be 
\begin{equation*}
\textbf{R} = \left[\begin{array}{cc}
\textbf{R}_{1} & \textbf{E}\\
\textbf{0} & \textbf{R}_{2}\\
\end{array}\right],
\end{equation*}
where $\textbf{R}_{1}$ is a $L \times L$ block-orthogonal matrix with parameters $\left(\Gamma-1, k, \gamma\right)$, $\textbf{E}$ is an $L \times l$ matrix and $\textbf{R}_{2}$ is a $l \times l$ upper triangular matrix. The STBC will be a block orthogonal STBC with parameters $\left(\Gamma, k, \gamma\right)$ if the following conditions are satisfied: 
\begin{itemize}
\item The matrices $\left\lbrace \textbf{B}_{1} , ... , \textbf{B}_{l}\right\rbrace $ are $k$-group decodable with $\gamma$ variables in each group, i.e., $\left\lbrace \textbf{B}_{1} , ... , \textbf{B}_{l}\right\rbrace $ can be partitioned into $k$ sets $ \left\lbrace \textbf{S}_{1}, ..., \textbf{S}_{k}\right\rbrace $, each of cardinality $\gamma$ such that $ \textbf{B}_{i} \textbf{B}_{j}^{H} + \textbf{B}_{j} \textbf{B}_{i}^{H} = \textbf{0}$ for all $ \textbf{B}_{i} \in \textbf{S}_{m}$, $ \textbf{B}_{j} \in \textbf{S}_{n}$, $m \neq n$.  
\item The set of matrices $ \left\lbrace \textbf{A}_{1},..., \textbf{A}_{L}, \textbf{B}_{1},..., \textbf{B}_{l}\right\rbrace $ are such that the $ \textbf{R}$ matrix obtained has full rank.  
\item The matrix $\textbf{E}^{H}\textbf{E}$ is a block diagonal matrix with $k$ blocks of size $\gamma \times \gamma$.
\end{itemize}
\end{lemma}
\begin{proof}
Proof is given in Appendix \ref{proof_bostc_lemma4}.
\end{proof}

\subsection{Effect of ordering on block orthogonality}
\label{bostbc_ordering}

We now show that the block orthogonality property depends on the ordering of the weight matrices or equivalently the ordering of the variables. If we do not choose the right ordering, we will be unable to get the desired structure. 
\begin{example}
\label{golden_code}
Let us consider the Golden code \cite{BRV} given by:
\begin{equation}
\label{golden_code_eq}
\textbf{X} = \frac{1}{\sqrt{5}} \left[\begin{array}{rr}
\alpha\left(s_{1} + s_{2} \theta \right)  & j\overline{\alpha}\left(s_{3} + s_{4} \overline{\theta} \right)\\
\alpha\left(s_{3} + s_{4} \theta \right) & \overline{\alpha}\left(s_{1} + s_{2} \overline{\theta} \right)\\
\end{array}\right],
\end{equation}
where $\theta = \left( 1 + \sqrt{5}\right) /2$, $\overline{\theta} = \left( 1 - \sqrt{5}\right) /2$, $\alpha = 1 + j\left(1 - \theta\right) $, $\overline{\alpha} = 1 + j\left(1 - \overline{\theta}\right) $ and $s_{i} = s_{iI} + j s_{iQ}$ for $i = 1,...,4$. 

If we order the variables (and hence the weight matrices) as  $\left[s_{1I} , s_{1Q} , s_{2I} , s_{2Q} , s_{3I} , s_{3Q} , s_{4I} , s_{4Q} \right]$, then the $ \textbf{R}$ matrix for SD has the following structure
\begin{equation*}
\textbf{R} = \left[\begin{array}{cccccccc}
\textbf{t} & \textbf{0} & 0 & t & t & t & t & t\\
0 & \textbf{t} & t & 0 & t & t & t & t\\
0 & 0 & \textbf{t} & \textbf{0} & t & t & t & t\\
0 & 0 & 0 & \textbf{t} & t & t & t & t\\
0 & 0 & 0 & 0 & \textbf{t} & \textbf{0} & 0 & t\\
0 & 0 & 0 & 0 & 0 & \textbf{t} & t & 0\\
0 & 0 & 0 & 0 & 0 & 0 & \textbf{t} & \textbf{0}\\
0 & 0 & 0 & 0 & 0 & 0 & 0 & \textbf{t}\\
\end{array}\right],
\end{equation*}
where $t$ denotes non zero entries. This ordering of variables has presented a $ \left( 4, 2, 1\right) $ block orthogonal structure to the $\textbf{R}$ matrix. Now, if we change the ordering to $\left[s_{1I} , s_{2I} , s_{1Q} , s_{2Q} , s_{3I} , s_{4I} , s_{3Q} , s_{4Q} \right]$, then the $ \textbf{R}$ matrix for SD has the following structure
\begin{equation*}
\textbf{R} = \left[\begin{array}{cccccccc}
\textbf{t} & \textbf{t} & \textbf{0} & \textbf{0} & t & t & t & t\\
0 & \textbf{t} & \textbf{0} & \textbf{0} & t & t & t & t\\
0 & 0 & \textbf{t} & \textbf{t} & t & t & t & t\\
0 & 0 & 0 & \textbf{t} & t & t & t & t\\
0 & 0 & 0 & 0 & \textbf{t} & \textbf{t} & \textbf{0} & \textbf{0}\\
0 & 0 & 0 & 0 & 0 & \textbf{t} & \textbf{0} & \textbf{0}\\
0 & 0 & 0 & 0 & 0 & 0 & \textbf{t} & \textbf{t}\\
0 & 0 & 0 & 0 & 0 & 0 & 0 & \textbf{t}\\
\end{array}\right],
\end{equation*}
where $t$ denotes non zero entries. This ordering of variables has presented a $ \left( 2, 2, 2\right) $ block orthogonal structure to the $\textbf{R}$ matrix. 
We can also have an ordering which can leave the $\textbf{R}$ matrix bereft of any block orthogonal structure such as $\left[s_{1I} , s_{1Q} , s_{4I} , s_{2Q} , s_{3I} , s_{3Q} , s_{2I} , s_{4Q} \right]$. The structure of the $\textbf{R}$ matrix in this case will be
\begin{equation*}
\textbf{R} = \left[\begin{array}{cccccccc}
t & 0 & t & 0 & t & t & t & t\\
0 & t & t & t & t & t & 0 & t\\
0 & 0 & t & t & t & t & t & 0\\
0 & 0 & 0 & t & t & t & t & t\\
0 & 0 & 0 & 0 & t & t & t & t\\
0 & 0 & 0 & 0 & 0 & t & t & t\\
0 & 0 & 0 & 0 & 0 & 0 & t & t\\
0 & 0 & 0 & 0 & 0 & 0 & 0 & t\\
\end{array}\right],
\end{equation*}
Also note that we have many entries $r_{ij} \neq 0$ even when the $i$-th and the $j$-th weight matrices are HR orthogonal such as for cases $i=6,j=8$ and $i=5,j=8$ etc. 
\end{example}

\section{Construction of Block Orthogonal STBCs}
\label{sec4}

Code constructions for block orthogonal STBCs with various parameters were presented in \cite{RGYZ}. It was shown via simulations that these constructions were indeed block orthogonal with the aforementioned parameters. We provide analytical proofs for the block orthogonal structure of some of these constructions which include also other well known codes such as the BHV code \cite{BHV}, the Silver code \cite{HLRVV} and the Srinath-Rajan code \cite{PaR}. We first study some basics of CUWDs and CIODs. 

\subsection{CUWDs and CIODs}
\label{cuwd_ciod}

\subsubsection{CUWDs}
\label{cuwds}

\cite{RaR} Linear STBCs can be broadly classified as unitary weight designs (UWDs) and non unitary weight designs (NUWDs). A UWD is one for which all the weight matrices are unitary and NUWDs are defined as those which are not UWDs. Clifford unitary weight designs (CUWDs) are a proper subclass of UWDs whose weight matrices satisfy certain sufficient conditions for $g$-group ML decodability. To state those sufficient conditions, let us list down the weight matrices of a CUWD in the form of an array as shown in Table \ref{cuwd_table}. 

\begin{table}[ht]
\caption{Structure of CUWDs}
\centering
\begin{tabular}{cccc}
\hline 
$\textbf{A}_{1}$ & $\textbf{A}_{\lambda + 1}$ & $\cdots$ & $\textbf{A}_{\left( g-1\right) \lambda + 1}$\\
$\textbf{A}_{2}$ & $\textbf{A}_{\lambda + 2}$ & $\cdots$ & $\textbf{A}_{\left( g-1\right) \lambda + 2}$\\
$\vdots$ & $\vdots$ & $\ddots$ & $\vdots$\\
$\textbf{A}_{\lambda}$ & $\textbf{A}_{2\lambda}$ & $\cdots$ & $\textbf{A}_{K}$\\
\hline
\end{tabular}
\label{cuwd_table}
\end{table}

All the weight matrices in one column belong to one group. The weight matrices of CUWDs satisfy the following sufficient conditions for $g$-group ML decodability.

\begin{itemize}
\item $\textbf{A}_{1} = \textbf{I}$.
\item All the matrices in the first row except $\textbf{A}_{1}$ should square to $−\textbf{I}$ and should pair-wise anti-commute among themselves.
\item The unitary matrix in the $i$-th row and the $j$-th column is equal to $\textbf{A}_{i} \textbf{A}_{\left( j−1\right) λ+1}$.
\end{itemize}

The CUWD matrix representation for these matrices for a system with $2^{a}$ transmit antennas are given below \cite{KaR1}. Let 
\begin{equation*}
\sigma_{1} = \left[\begin{array}{cc}
0 & 1\\
-1 & 0\\
\end{array}\right], ~~ 
\sigma_{2} = \left[\begin{array}{cc}
0 & j\\
j & 0\\
\end{array}\right], ~~
\sigma_{3} = \left[\begin{array}{cc}
1 & 0\\
0 & -1\\
\end{array}\right].
\end{equation*}

The representations of the Clifford generators are given by:
\begin{equation*}
R\left( \gamma_{1}\right) = \pm j \sigma_{3}^{\otimes^{a}},
\end{equation*}
\begin{equation*}
R\left( \gamma_{2k}\right) = \textbf{I}_{2}^{\otimes^{a-k}}\bigotimes \sigma_{1} \bigotimes \sigma_{3}^{\otimes^{k-1}},
\end{equation*}
\begin{equation*}
R\left( \gamma_{2k+1}\right) = \textbf{I}_{2}^{\otimes^{a-k}}\bigotimes \sigma_{2} \bigotimes \sigma_{3}^{\otimes^{k-1}},
\end{equation*}
\begin{equation*}
R\left( \gamma_{0}\right) = \textbf{I}_{2^{a}},
\end{equation*}
where $k = 1, ..., a$.
The weight matrices of the CUWD for a rate-1, four group decodable STBC can be derived as follows. Let $\alpha_{i} = jR\left( \gamma_{2i}\right) R\left( \gamma_{2i+1}\right) $ for $i = 1, 2, ..., a-1$. Let $ \lambda = 2^{a-1}$. The weight matrices are now given by  
\begin{equation}
\label{cuwd_mat}
\begin{aligned}
\textbf{A}_{ \lambda + 1} = R\left(1\right)  ,\\
\textbf{A}_{ 2\lambda + 1} = R\left( \gamma_{2a+1}\right) ,\\
\textbf{A}_{3 \lambda + 1} = R\left( \gamma_{2a}\right) ,\\
\textbf{A}_{j \lambda + k} = \textbf{A}_{k} \textbf{A}_{j \lambda + 1}, \\
\textbf{A}_{k} = \prod_{i=1}^{a-1} \alpha_{i}^{k_{i}}
\end{aligned}
\end{equation}
for $j = 1,2,3$, $k = 1,.. \lambda$ and where $\left( k_{1}, k_{2}, ..., k_{a-1}\right)  $ is the binary representation of $k-1$. 

\subsubsection{CIODs}
\label{ciods}

Coordinate interleaved orthogonal designs (CIODs) were introduced in \cite{KhR}. 
\begin{definition}
\label{ciod_def}
A CIOD for a system with $2^{a}$ transmit antennas in variables $x_{i}$, $i = 1,...,K-1$, $K$ even, is a $2^{a} \times 2^{a}$ matrix $S\left( x_{0}, ..., x_{K-1}\right) $, such that
\begin{equation}
\label{ciod_mats}
S = \left[\begin{array}{cc}
\Theta_{1} \left(\tilde{x}_{0}, ..., \tilde{x}_{\frac{K}{2}-1}\right)  & \textbf{0}\\
\textbf{0} & \Theta_{2} \left(\tilde{x}_{\frac{K}{2}}, ..., \tilde{x}_{K-1}\right)\\
\end{array}\right],
\end{equation}
where $\Theta_{1} \left(\tilde{x}_{0}, ..., \tilde{x}_{\frac{K}{2}-1}\right)$ and $\Theta_{2} \left(\tilde{x}_{\frac{K}{2}}, ..., \tilde{x}_{K-1}\right)$ are complex orthogonal designs of size $2^{a-1} \times 2^{a-1}$ and $\tilde{x}_{i} = x_{iI} + jx_{\left( i + K/2\right) mod K}$. 
\end{definition}

\subsection{BOSTBCs from CUWDs}
\label{bostc_cuwd_ciod}

We now show that STBCs obtained as a sum of rate-1, four group decodable CUWDs exhibit the block orthogonal structure with parameters $\left(2, 4, \lambda\right)$. 

\begin{lemma}
\label{bostc_cuwd_const}
\textit{Construction I:} Let $ \textbf{X}_{1}\left( s_1, s_2, ..., s_{4\lambda}\right) $ be a rate-1, four group decodable STBC obtained from CUWD \cite{RaR} with weight matrices $ \left\lbrace \textbf{A}_{1}, \textbf{A}_{2}, ..., \textbf{A}_{4 \lambda}\right\rbrace $. Let $ \textbf{M}$ be an $n_t \times n_t$ matrix such that the set of weight matrices $ \left\lbrace \textbf{A}_{1}, \textbf{A}_{2}, ..., \textbf{A}_{4 \lambda}, \textbf{M}\textbf{A}_{1}, \textbf{M}\textbf{A}_{2},..., \textbf{M} \textbf{A}_{4 \lambda}\right\rbrace $ yield a full rank $ \textbf{R}$ matrix. Then the STBC given by
\begin{align*}
\textbf{X}\left( s_1, s_2, ..., s_{8 \lambda}\right)  &= \textbf{X}_{1}\left( s_1, s_2, ..., s_{4\lambda}\right) \\
&\quad  + \textbf{M}.\textbf{X}_{1}\left( s_{4 \lambda + 1}, s_{4 \lambda + 2 }, ..., s_{8\lambda}\right),
\end{align*}
will exhibit a block orthogonal structure with parameters $\left( 2, 4, \lambda\right) $.
\end{lemma}
\begin{proof}
Proof is given in Appendix \ref{app_r_mat_struct_const_1}.
\end{proof}

\begin{example}
\label{bostc_cuwd_ex}
Let us consider the BHV code given by:
\begin{equation*}
\textbf{X} = \textbf{X}_{1}\left(s_{1}, s_{2}\right) + \textbf{T}\textbf{X}_{1}\left(z_{1}, z_{2}\right),
\end{equation*}
where $\textbf{X}_{1}$ and $\textbf{X}_{1}$ take the Alamouti structure, and
$\textbf{X}_{1}\left(s_{1}, s_{2}\right) = \left[\begin{array}{rr}
s_{1} & -s_{2}^{*}\\
s_{2} & s_{1}^{*}\\
\end{array}\right],$
$\textbf{T} = \left[\begin{array}{cc}
1 & 0\\
0 & -1\\
\end{array}\right]$ 
and $\left[ z_{1}, z_{2}\right]^{T} = \textbf{U}\left[ s_{3}, s_{4}\right]^{T},$ where $\textbf{U}$ is a unitary matrix chosen to maximize the minimum determinant. In this case, as per the above construction, $ \textbf{M} = \textbf{T} \textbf{U}$. Hence, the BHV code is a BOSTBC with parameters $ \left( 2,4,1\right) $.
\end{example}

\subsection{BOSTBCs from Cyclic Division / Crossed Product Algebras}
\label{bostc_cda_cpa}
In this section, we show the block orthogonality property of two constructions from either cyclic division algebras or crossed product algebras over the field $ \mathbb{Q}\left( i\right) $. 

\begin{lemma}
\label{bostc_const_2}
\textit{Construction II:} Let $\textbf{X}$ be an STBC with weight matrices $\left\lbrace \textbf{A}_{1}, ... , \textbf{A}_{K}\right\rbrace $ and $\left\lbrace \textbf{B}_{1}, ... , \textbf{B}_{K}\right\rbrace $ for the variables $\left[ x_{1I} ~ ... ~ x_{KI}\right] $ and $\left[ x_{1Q} ~ ... ~ x_{KQ}\right] $ respectively such that $\textbf{B}_{i} = j\textbf{A}_{i}$ for $1 \leq i \leq K$. Let the weight matrices be chosen such that the $ \textbf{R}$ matrix has full rank. Then the code $\textbf{X}$ exhibits the block orthogonal property with parameters $\left( K, 2, 1\right) $ if we take the ordering of weight matrices as $\left\lbrace \textbf{A}_{1}, \textbf{B}_{1},..., \textbf{A}_{K}, \textbf{B}_{K}\right\rbrace $. 
\end{lemma}
\begin{proof}
Proof is given in Appendix \ref{app_r_mat_struct_const_2}.
\end{proof}

\begin{example}
\label{bostc_cda_ex}
Consider any STBC obtained from the Cyclic Division Algebra (CDA) \cite{SeRS} over the base field $ \mathbb{Q}\left( i\right) $. The structure of such an STBC will be 
\begin{equation*}
\textbf{X} = \left[\begin{array}{cccc}
x_{0} & \gamma \sigma\left( x_{n-1}\right) & \cdots & \gamma \sigma^{n-1}\left( x_{1}\right)\\
x_{1} & \sigma\left( x_{0}\right) & \cdots & \gamma \sigma^{n-1}\left( x_{n-2}\right)\\
\vdots & \vdots & \ddots & \vdots\\
x_{n-1} & \sigma\left( x_{n-2}\right) & \cdots &  \sigma^{n-1}\left( x_{0}\right)
\end{array}\right],
\end{equation*}
where $x_{k} = x_{kI} + jx_{kQ}$. The weight matrices of this STBC satisfy the properties of the construction above. Hence, this is a BOSTBC with parameters $\left( n, 2, 1\right) $. 
\end{example}

The next construction is a special case of the previous construction. 
\begin{lemma}
\label{bostc_const_3}
\textit{Construction III:} Let $\textbf{X}_{1}$ be a two group decodable STBC with weight matrices $\left\lbrace \textbf{A}_{1}, ... , \textbf{A}_{K}\right\rbrace $ and $\left\lbrace \textbf{B}_{1}, ... , \textbf{B}_{K}\right\rbrace $ for the variables $\left[ x_{1} ~ ... ~ x_{K}\right] $ and $\left[ x_{K+1} ~ ... ~ x_{2K}\right] $ respectively such that $\textbf{B}_{i} = j\textbf{A}_{i}$ for $i= 1,...K$. Let $ \textbf{M}$ be a matrix such that the set of weight matrices $ \left\lbrace \textbf{A}_{1}, \textbf{A}_{2}, ..., \textbf{A}_{K}, \textbf{M} \textbf{A}_{1}, \textbf{M} \textbf{A}_{2}, ..., \textbf{M} \textbf{A}_{K}\right\rbrace $ yield a full rank $ \textbf{R}$ matrix. Then the STBC given by
\begin{align*}
\textbf{X}\left( x_1, x_2, ..., x_{4K}\right)  &= \textbf{X}_{1}\left( x_1, x_2, ..., x_{2K}\right) \\
&\quad+ \textbf{M}.\textbf{X}_{1}\left( x_{2K + 1}, x_{2K + 2 }, ..., x_{4K}\right),
\end{align*}
will exhibit a block orthogonal structure with parameters $\left( 2, 2, K\right) $.
\end{lemma}
\begin{proof}
Proof is given in Appendix \ref{app_r_mat_struct_const_3}.
\end{proof}

\begin{example}
\label{bostc_cda_2g_ex}
Consider the golden code as given in example \ref{golden_code}. If we consider, 
\begin{equation*}
\textbf{X}_{1} = \frac{1}{\sqrt{5}} \left[\begin{array}{rr}
\alpha\left(s_{1} + s_{2} \theta \right)  & 0\\
0 & \overline{\alpha}\left(s_{1} + s_{2} \overline{\theta} \right)\\
\end{array}\right],
\end{equation*}
and $ \textbf{M}$ as 
\begin{equation*}
\textbf{M} = \left[\begin{array}{rr}
0 & j\\
1 & 0\\
\end{array}\right],
\end{equation*}
we can see that the golden code is a BOSTBC with parameters $\left(2,2,2\right) $. 
\end{example}

\subsection{BOSTBCs from CIODs}
\label{bostc_ciod}
In this section we show that the BOSTBCs that can be obtained from CIODs \cite{KhR}. 
\begin{lemma}
\label{bostc_ciod_const}
\textit{Construction IV:} Let $ \textbf{X}_{1}\left( s_1, s_2, ..., s_{K}\right) $ be a rate-1 CIOD with weight matrices $ \left\lbrace \textbf{A}_{1}, \textbf{A}_{2}, ..., \textbf{A}_{K}\right\rbrace $. Let $ \textbf{M}$ be a matrix such that the set of weight matrices $ \left\lbrace \textbf{A}_{1}, \textbf{A}_{2}, ..., \textbf{A}_{K}, \textbf{M} \textbf{A}_{1}, \textbf{M} \textbf{A}_{2}, ...,  \textbf{M} \textbf{A}_{K}\right\rbrace $ yield a full rank $ \textbf{R}$ matrix. Then the STBC given by
\begin{align*}
\textbf{X}\left( s_1, s_2, ..., s_{2K}\right)  &= \textbf{X}_{1}\left( s_1, s_2, ..., s_{K}\right) \\
&\quad + \textbf{M}\textbf{X}_{1}\left( s_{K + 1}, s_{K + 2 }, ..., s_{2K}\right),
\end{align*}
will exhibit a block orthogonal structure with parameters $\left( 2, K/2, 2\right) $.
\end{lemma}
\begin{proof}
Proof is given in Appendix \ref{app_r_mat_struct_const_4}.
\end{proof}

\begin{example}
\label{bostc_ciod_ex}
Consider the $2 \times 2$ code constructed by Srinath et al. in \cite{PaR} given by
\begin{equation*}
\textbf{X} = \left[\begin{array}{rr}
x_{1I} + jx_{2Q}  & e^{j\pi/4}\left( x_{3I} + jx_{4Q}\right) \\
e^{j\pi/4}\left( x_{4I} + jx_{3Q}\right) & x_{2I} + jx_{1Q}\\
\end{array}\right],
\end{equation*}
If we consider, 
\begin{equation*}
\textbf{X}_{1} = \left[\begin{array}{rr}
x_{1I} + jx_{2Q}  & 0\\
0 & x_{2I} + jx_{1Q}\\
\end{array}\right],
\end{equation*}
and $ \textbf{M}$ as 
\begin{equation*}
\textbf{M} = \left[\begin{array}{rr}
0 & e^{j\pi/4}\\
e^{j\pi/4} & 0\\
\end{array}\right],
\end{equation*}
we see that the code is a BOSTBC with parameters $\left(2,2,2\right) $. 
\end{example}


\section{Reduction of Decoding complexity for Block Orthogonal Codes}
\label{sec5}

In this section we describe how we can achieve decoding complexity reduction for BOSTBCs. Also we show how the block orthogonal structure helps in the reduction of the Euclidean Metric (EM) calculations and the sorting operations for a sphere decoder using a depth first search algorithm. We also briefly present the implications of the block orthogonal structure for QRDM decoders as discussed in \cite{RGYZ}. 

\subsection{ML decoding complexity reduction}
\label{ml_dec_cplxity_redn}
The sphere decoder under consideration in this section will be the depth first search algorithm based decoder with Schnorr-Euchner enumeration and pruning as discussed in \cite{SiB}. We first consider the case of $\Gamma = 2$ Block Orthogonal Code. 

\subsubsection{$\Gamma = 2$}
\label{ml_gamma_2}
Consider a BOSTBC with parameters $\left( 2, k, \gamma\right) $. The structure of the $\textbf{R}$ matrix for this code is as mentioned in \eqref{block_ortho_r_mat} with two blocks $ \textbf{R}_{1}$ and $ \textbf{R}_{2}$. This code is fast sphere decodable, i.e., for a given set of values of variables in sub-blocks $ \textbf{U}_{2,j}$, $j=1,...,k$, we can decode the variables in $ \textbf{U}_{1,j}$ and $ \textbf{U}_{1,l}$, $1 \leq j < l \leq k$, independently.  The ML decoding complexity of this code will be $ O\left(M^{k \gamma + \gamma}\right) $. Due to the structure of the block orthogonal code, we can see that the variables in the blocks $ \textbf{U}_{2,j}$ and $ \textbf{U}_{2,l}$, $1 \leq j < l \leq k$, are also independent in the sense that the EM calculations and the Schnorr-Euchner enumeration based sorting operations for the variables in $ \textbf{U}_{2,j}$ are independent of the values taken by the variables in $ \textbf{U}_{2,l}$. We illustrate this point with an example.
\begin{example}
\label{bostbc_em_ind_ex}
Consider a hypothetical BOSTBC having the parameters $\left( 2, 2, 1\right) $ with variables $\left\lbrace x_{1}, x_{2}, x_{3}, x_{4} \right\rbrace $. The $ \textbf{R}$ matrix for this BOSTBC will be of the form
\begin{equation*}
\textbf{R} = \left[\begin{array}{cccc}
\textbf{t} & \textbf{0} & t & t\\
0 & \textbf{t} & t & t\\
0 & 0 & \textbf{t} & \textbf{0}\\
0 & 0 & 0 & \textbf{t}\\
\end{array}\right]
\end{equation*}
The first two levels of the search tree for the sphere decoder are shown in in Figure \ref{fig:bostbc_em_ind} with the variables assumed to be taking values from a 2-PAM constellation - A. As it can be seen from the figure, irrespective of the value taken by $x_{4}$, the edge weights (Euclidean metrics) for the variable $x_{3}$ remain the same. 
\end{example}
\begin{figure}
\centering
\includegraphics[width=6.5in,height=2.7in]{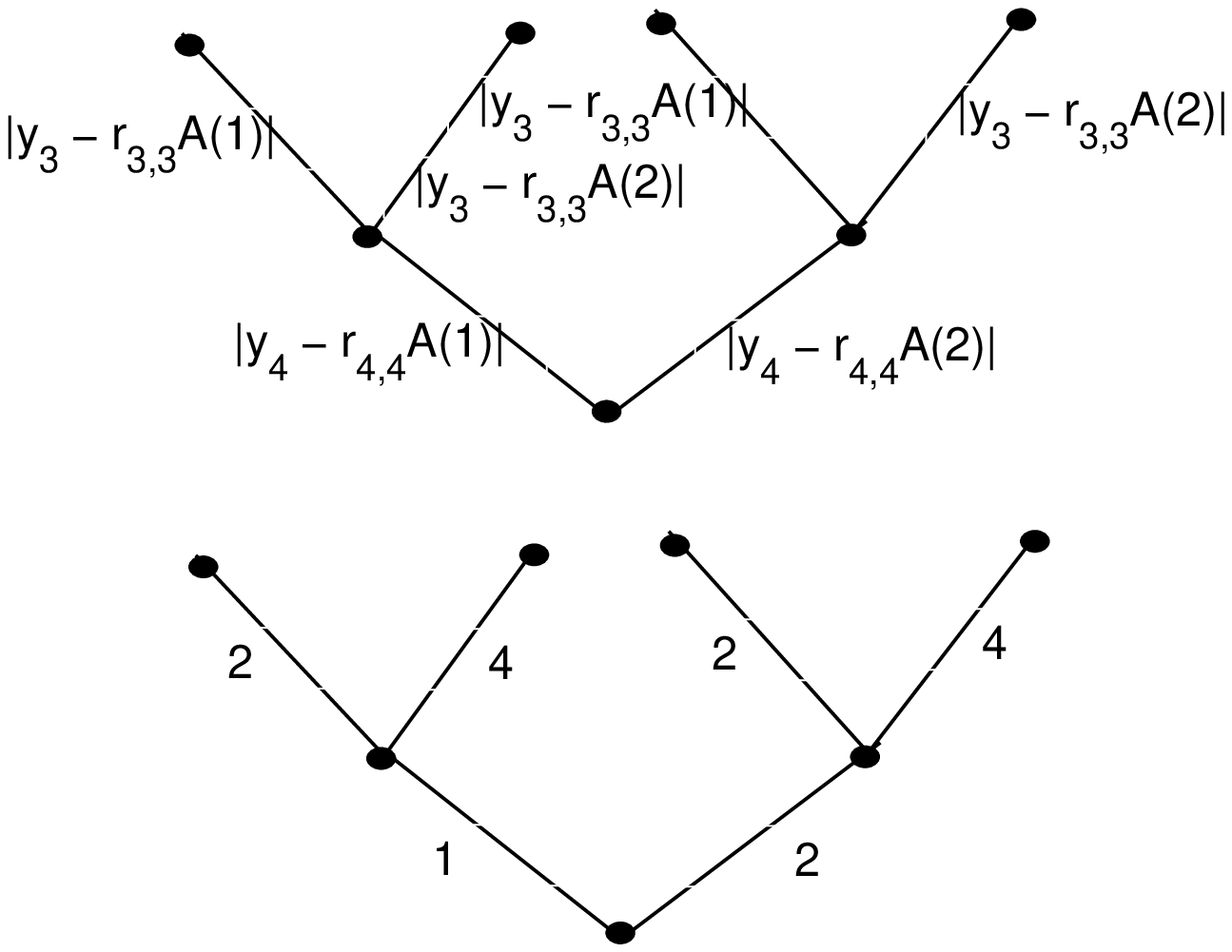}
\caption{First two levels of the sphere decoder tree for the code in example \eqref{bostbc_em_ind_ex}} 
\label{fig:bostbc_em_ind}
\end{figure}
From example \ref{bostbc_em_ind_ex} we can see that instead of calculating the EM repeatedly, we can store these values in a look up table when they are calculated for the first time and retrieve them whenever needed. This technique of avoiding repeated calculations by storing the previously calculated values is known as \textit{Memoization} \cite{CLRS}. This approach reduces the number of floating point operations (FLOPS) significantly. 

\subsubsection{$\Gamma > 2$}
\label{ml_gamma_g2}
Consider a BOSTBC with parameters $\left( \Gamma, k, \gamma\right) $. The structure of the $\textbf{R}$ matrix for this code is as mentioned in \eqref{block_ortho_r_mat}. Consider the block $ \textbf{R}_{i} $, $1 < i \leq \Gamma$ of the $ \textbf{R}$ matrix. For a given set of values for the variables in the blocks $ \textbf{R}_{m}$, $m > i$, we can see that the variables in the blocks $ \textbf{U}_{i,j}$ and $ \textbf{U}_{i,l}$, $1 \leq j < l \leq k$, are independent as seen in the case of $ \Gamma = 2$. Hence, we can use memoization here as well in order to reduce the number of EM calculations and sorting operations.  

\subsection{Complexity reduction bound and Memory requirements for depth first sphere decoder}
\label{cplxity_redn_mem_req}
We calculate the maximum possible reduction in the number of EM values calculated and the memory requirements for the look up tables in this section. First we consider the case of $ \Gamma = 2$. 

\subsubsection{$ \Gamma = 2$}
\label{bound_gamma_2}
Considering a $ \left( 2, k, \gamma\right) $ BOSTBC, we first calculate the memory requirements for storing the EM values. Let each of the variables of the STBC take values from a constellation of size $M$. The number of EM values that need to be stored for a single sub-block $ \textbf{U}_{2,j}$, $ 1 \leq j < k$, is 
\begin{align*}
Mem\left( \textbf{U}_{2,j}\right) &= M + M^{2} + ... + M^{\gamma}\\
&= \frac{M^{\gamma + 1} - M}{M - 1} = \frac{M\left( M^{ \gamma} - 1\right) }{M - 1}.
\end{align*}
These values will need to be stored for $\left( k-1\right) $ such sub-blocks. The total memory requirement for the block $ \textbf{R}_{2}$ is,
\begin{equation*}
Mem\left( \textbf{R}_{2}\right) = \left( k-1\right) \frac{M\left( M^{ \gamma} - 1\right) }{M - 1}.
\end{equation*}

We now find the maximum number of reductions possible for the EM calculations for this BOSTBC. This will occur when all the nodes are visited in the depth first search. For the block $ \textbf{R}_{2}$, the number of EM calculations for a code without the block orthogonal structure would be 
\begin{align*}
O_{STBC} &= M + M^{2} + ... + M^{k \gamma} = \frac{M\left( M^{k \gamma} - 1\right) }{M - 1}.
\end{align*}
For a BOSTBC, if we use the look up table, we would be performing the EM calculations only once per each of the sub-block. For $k$ sub-blocks, the number of EM calculations will be 
\begin{align*}
O_{BOSTBC} &= k\left(  M + M^{2} + ... + M^{ \gamma}\right) \\
&= k\frac{M\left( M^{ \gamma} - 1\right) }{M - 1}.
\end{align*}
We therefore perform only a small percentage of EM calculations if the code exhibits a block orthogonal structure. We call the ratio of the the number of EM calculated for a BOSTBC to the number of EM calculated if the STBC did not possess a block orthogonal structure as Euclidean Metric Reduction Ratio (EMRR) given by
\begin{align*}
\frac{O_{BOSTBC}}{O_{STBC}} &= \frac{k\frac{M\left( M^{ \gamma} - 1\right) }{M - 1}}{\frac{M\left( M^{k \gamma} - 1\right) }{M - 1}} = \frac{k\left( M^{ \gamma} - 1\right)}{\left( M^{k \gamma} - 1\right)}\\
&\approx \frac{k}{M^{\left( k-1\right) \gamma}},
\end{align*}
which is a decreasing function of $k$, $M$ and $ \gamma$. 

\subsubsection{$ \Gamma > 2$}
\label{bound_gamma_g2}
Considering a $ \left( \Gamma, k, \gamma\right) $ BOSTBC, we first calculate the memory requirements for storing the EM values. The memory requirement per sub-block $ \textbf{U}_{i,j}$, $ 1 \leq j < k$, of any block $ \textbf{R}_{i}$, $ 1 < j \leq \Gamma$, under consideration is the same as that of the case of the sub-block $ \textbf{U}_{2,j}$ in the $ \Gamma = 2$ case. This is so because, for a given set of values for the variables in the blocks $ \textbf{R}_{m}$, $i < m \leq \Gamma$, the memory requirement for the sub-block $ \textbf{U}_{i,j}$ can be calculated in the similar way as it was calculated for $ \textbf{U}_{2,j}$ for the $ \Gamma = 2$ case. Hence, the memory requirements for a block $ \textbf{R}_{i}$ for a given set of values for the variables in the blocks $ \textbf{R}_{m}$ is the same as that of $ \textbf{R}_{2}$ in the $ \Gamma = 2$ case. 
\begin{equation*}
Mem\left( \textbf{R}_{i}\right)_{conditional} = \left( k-1\right) \frac{M\left( M^{ \gamma} - 1\right) }{M - 1}.
\end{equation*}
We can reuse the same memory for another set of given values of the variables of $ \textbf{R}_{m}$, as the previous EM values will not be retrieved again as the depth first search algorithm does not revisit any of the previously visited nodes (i.e., any previously given set of values for the variables in the tree). Hence, we can write,
\begin{equation*}
Mem\left( \textbf{R}_{i}\right) = \left( k-1\right) \frac{M\left( M^{ \gamma} - 1\right) }{M - 1},
\end{equation*}
for $ 1 < i \leq \Gamma$. Since there are $ \Gamma - 1$ such blocks, the total memory requirement for storing the EM values will be
\begin{equation*}
Mem\left( \textbf{R}\right) = \left( \Gamma - 1\right) \left( k-1\right) \frac{M\left( M^{ \gamma} - 1\right) }{M - 1}.
\end{equation*}
We now find the maximum number of reductions possible for the EM calculations for this BOSTBC. This will occur when all the nodes are visited in the depth first search. For blocks other than $ \textbf{R}_{1}$, the number of EM calculations for a code without the block orthogonal structure would be 
\begin{align*}
O_{STBC} &= M + M^{2} + ... + M^{ \left( \Gamma - 1\right)  k \gamma}\\
&= \frac{M\left( M^{\left( \Gamma - 1\right)k \gamma} - 1\right) }{M - 1}.
\end{align*}
For a BOSTBC, if we consider the block $ \textbf{R}_{i}$ and for a given set of values for the variables in $ \textbf{R}_{m}$, $ i < m \leq \Gamma$, if we use the look up table, we would be performing the EM calculations only once per each of the sub-block. For $k$ sub-blocks, the number of EM calculations will be 
\begin{align*}
O_{BOSTBC}\left( \textbf{R}_{i}\right)_{conditional} &= k\left(  M + M^{2} + ... + M^{ \gamma}\right) \\
&= k\frac{M\left( M^{ \gamma} - 1\right) }{M - 1}.
\end{align*}
These calculations need to be repeated for all the $ M^{ \left( \Gamma - i\right) k \gamma }$ values of the variables in $ \textbf{R}_{m}$. 
\begin{align*}
O_{BOSTBC}\left( \textbf{R}_{i}\right) &= k M^{ \left( \Gamma - i\right) k \gamma }\left(  M + M^{2} + ... + M^{ \gamma}\right) \\
&= kM^{ \left( \Gamma - i\right) k \gamma }\frac{M\left( M^{ \gamma} - 1\right) }{M - 1}.
\end{align*}
The EM calculations for all the blocks is given by 
\begin{align*}
O_{BOSTBC} &= \sum_{i=2}^{\Gamma}kM^{ \left( \Gamma - i\right) k \gamma }\frac{M\left( M^{ \gamma} - 1\right) }{M - 1}\\
&= \frac{kM\left( M^{ \gamma} - 1\right)}{M-1}\frac{M^{ \left( \Gamma - 1\right) k \gamma } - 1}{M^{k \gamma } - 1}.
\end{align*}
The EMRR in this case will be
\begin{align*}
\frac{O_{BOSTBC}}{O_{STBC}} &= \frac{\frac{kM\left( M^{ \gamma} - 1\right)}{M-1}\frac{M^{ \left( \Gamma - 1\right) k \gamma } - 1}{M^{k \gamma } - 1}}{\frac{M\left( M^{\left( \Gamma - 1\right)k \gamma} - 1\right) }{M - 1}}\\
&= \frac{k\left( M^{ \gamma} - 1\right)}{\left( M^{k \gamma} - 1\right)} \approx \frac{k}{M^{\left( k-1\right) \gamma}}.
\end{align*}
We can see that the ratio of the reduction of operations is independent of $ \Gamma$ and dependent only on $k$ and $ \gamma$.

\subsection{QRDM decoding complexity reduction \cite{RGYZ}}
\label{qrdm_decoding_cplxity_redn}
In this section we review the simplified QRDM decoding method which exploits the block orthogonal structure of a code as presented in \cite{RGYZ}. The traditional QRDM decoder is a breadth first search decoder in which $M_{c}$ surviving paths with the smallest Euclidean metrics are picked at each stage and the rest of the paths are discarded. If $M_{c} = M^{\left( \Gamma - 1\right) k \gamma}$ for a block orthogonal code with parameters $\left( \Gamma, k, \gamma\right) $, then the QRDM decoder gives ML performance. The simplified QRDM decoder utilizes the block orthogonal structure of the code to find virtual paths between nodes, which reduces the number of surviving paths to effectively $M_{c_{eq}}$, to reduce the number of Euclidean metric calculations. For details of how this is achieved, refer to \cite{RGYZ}. The maximum reduction in decoding complexity bound for a QRDM decoder is given by
\begin{equation*}
\frac{O_{BOSTBC}}{O_{STBC}} = \frac{M^{\gamma}}{k\left( M^{\gamma} - 1\right) }.
\end{equation*}

\section{Simulation Results and Discussion}
\label{sec6}
In all the simulation scenarios in this section, we consider quasi-static Rayleigh flat fading channels and the channel state information (CSI) is known at the receiver perfectly. Any STBC which does not have a block orthogonal property is assumed to be a fast decodable STBC which is conditionally $k$ group decodable with $ \gamma$ symbols per group, but not possessing the block diagonal structure for the blocks $ \textbf{R}_{2}, ..., \textbf{R}_{\Gamma}$. 

\subsection{Sphere decoding using depth first search}
\label{ml_simulations}
We first plot the EMRR for BOSTBCs with different parameters against the SNR. Figures \ref{fig:stbc_241_em} and \ref{fig:stbc_222_em} show the plot of $O_{BOSTBC} / O_{STBC}$ vs SNR for a $ \left( 2,4,1\right) $ BOSTBC (examples - Silver code, BHV code) with the symbols being drawn from 4-QAM, 16-QAM and 64-QAM. We can clearly see that the reduction in the EMRR with the increasing size of signal constellation as explained in section \ref{cplxity_redn_mem_req}. It can also be seen that a larger value of $k$ gives a lower EMRR if we keep the product $k \gamma$ constant. Figure \ref{fig:stbc_242_em} shows the plot of $O_{BOSTBC} / O_{STBC}$ vs SNR for a $ \left( 2,4,2\right) $ BOSTBC (examples - $4 \times 2$ code from Pavan et al \cite{PaR}) with the symbols being drawn from 4-QAM and 16-QAM. Notice that the $ \left( 2,4,2\right) $ BOSTBC offers a lower EMRR as compared to the $\left( 2,4,1\right) $ BOSTBC due to the higher value of $ \gamma$, as explained in section \ref{cplxity_redn_mem_req}.

We now compare the total number of FLOPS performed by the sphere decoder for a BOSTBC against that of an STBC without a block orthogonal structure for various SNRs. Figures \ref{fig:stbc_241_flops}, \ref{fig:stbc_222_flops}, \ref{fig:stbc_242_flops} show the plot of number of FLOPS vs SNR for a $ \left( 2,4,1\right) $ BOSTBC, a $ \left( 2,2,2\right) $ BOSTBC and a  $ \left( 2,4,2\right) $ BOSTBC  respectively with the symbols being drawn from 4-QAM, 16-QAM and 64-QAM for the first two figures and from 4-QAM and 16-QAM for the last one. We can see that the BOSTBCs offer around 30\% reduction in the number of FLOPS for the  $ \left( 2,4,1\right) $ and  $ \left( 2,4,2\right) $  BOSTBCs and around  15\% for the $ \left( 2,2,2\right) $ BOSTBC at low SNRs. 

\begin{figure}
\centering
\includegraphics[width=6.5in,height=2.7in]{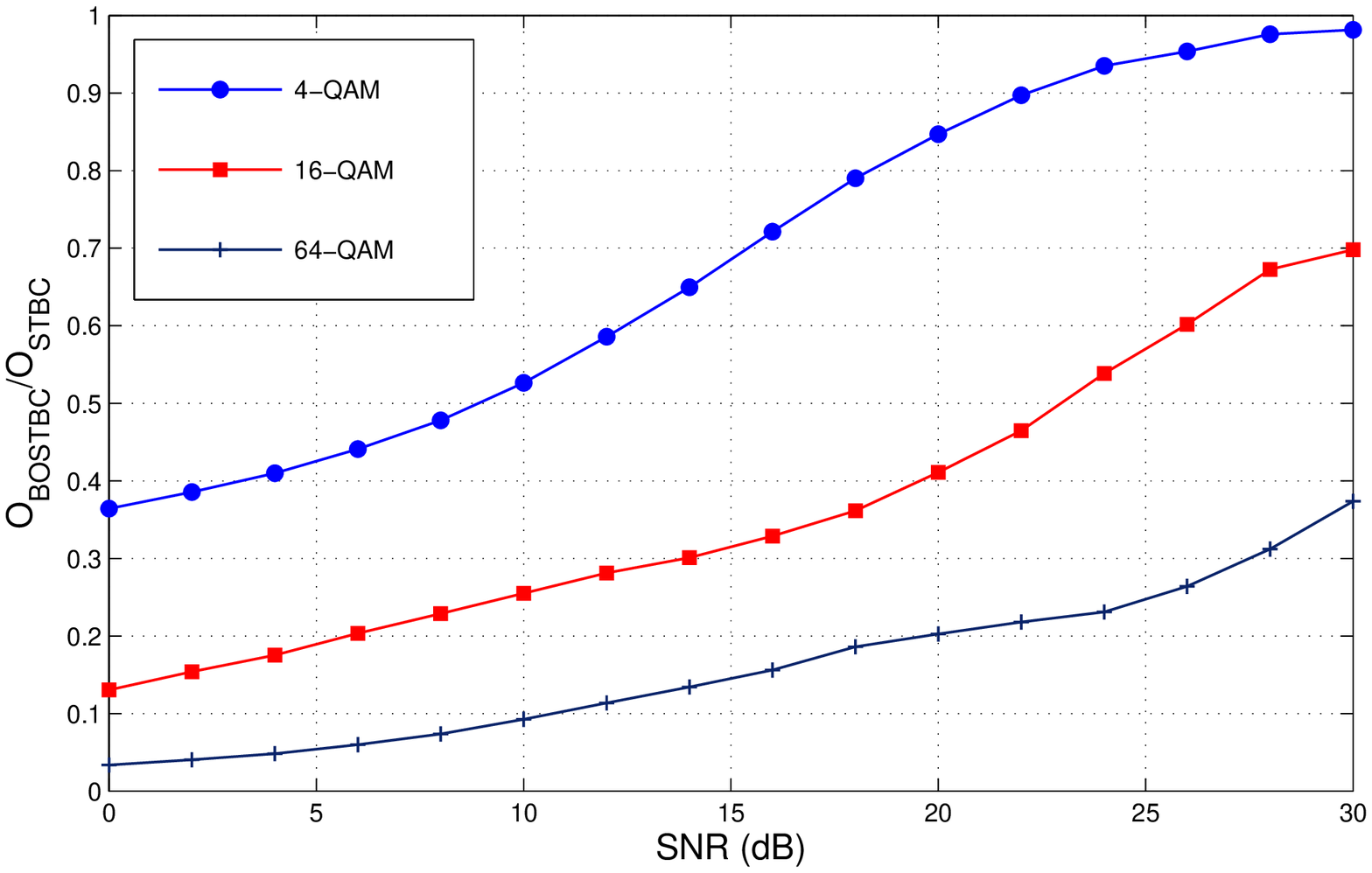}
\caption{The number of Euclidean metrics calculated ratio $ O_{BOSTBC} / O_{STBC} $ for a BOSTBC with parameters $\left( 2,4,1\right) $} 
\label{fig:stbc_241_em}
\end{figure}

\begin{figure}
\centering
\includegraphics[width=6.5in,height=2.7in]{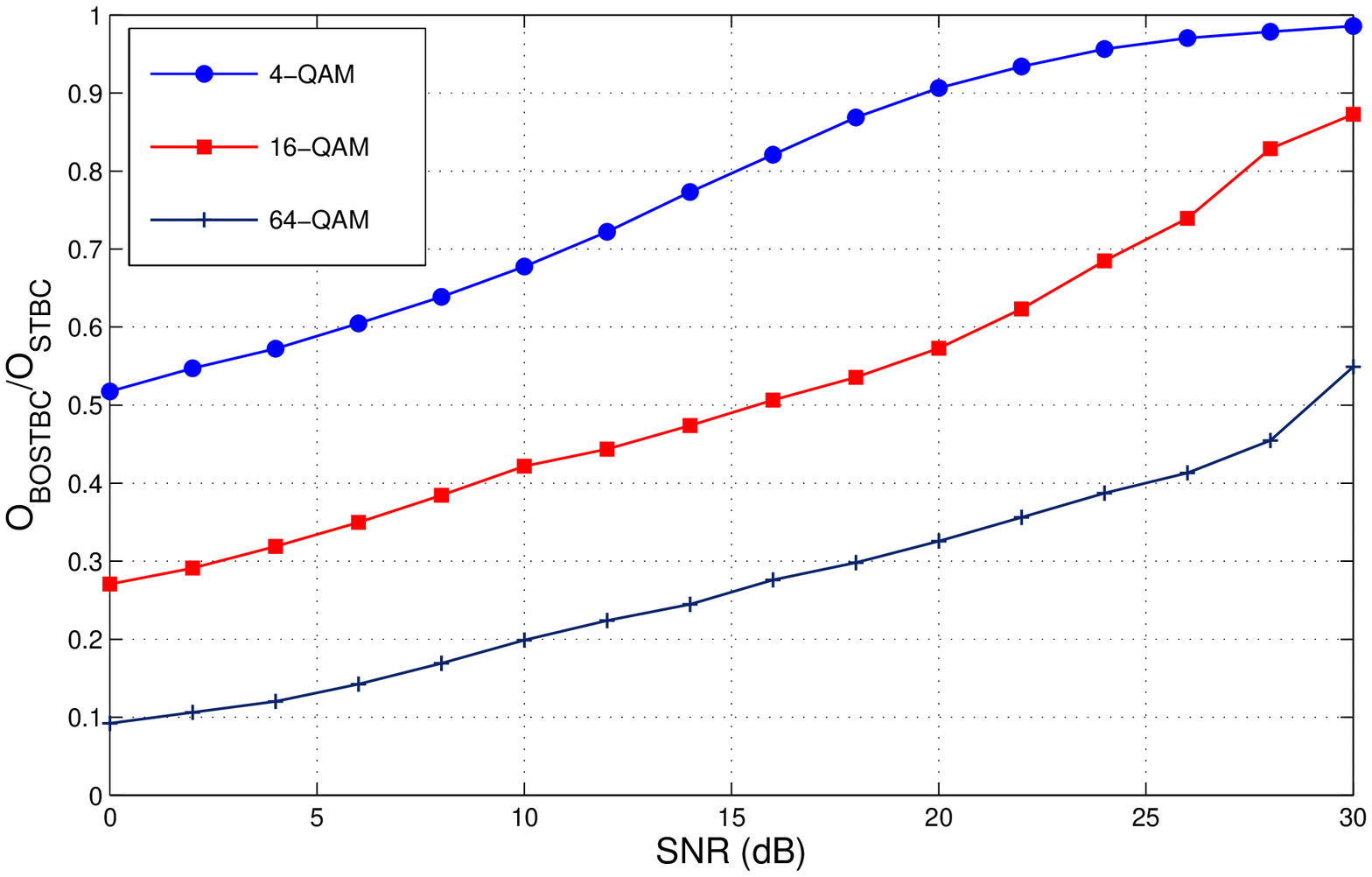}
\caption{The number of Euclidean metrics calculated ratio $ O_{BOSTBC} / O_{STBC} $ for a BOSTBC with parameters $\left( 2,2,2\right) $} 
\label{fig:stbc_222_em}
\end{figure}

\begin{figure}
\centering
\includegraphics[width=6.5in,height=2.7in]{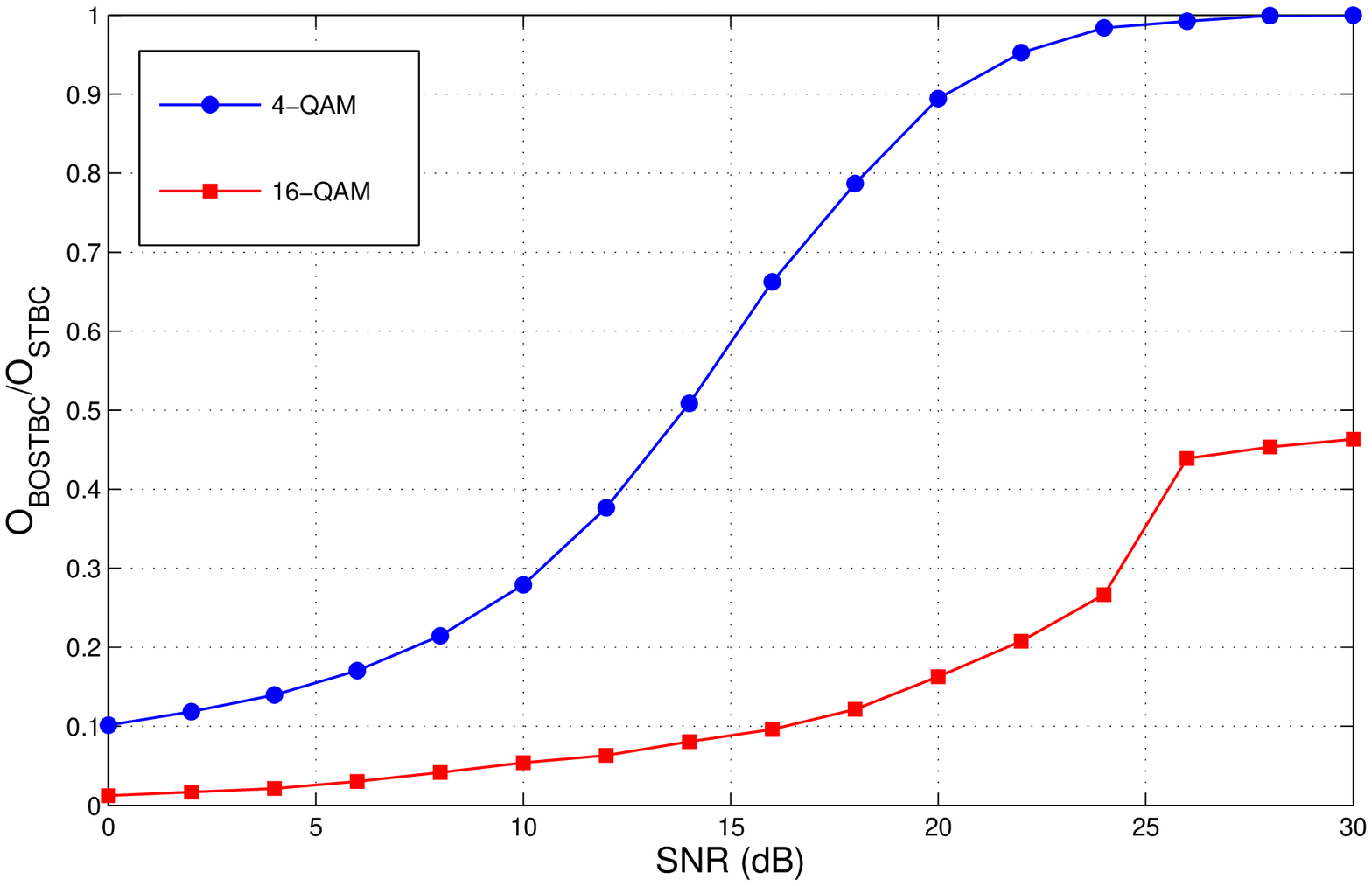}
\caption{The number of Euclidean metrics calculated ratio $ O_{BOSTBC} / O_{STBC} $ for a BOSTBC with parameters $\left( 2,4,2\right) $} 
\label{fig:stbc_242_em}
\end{figure}

\begin{figure}
\centering
\includegraphics[width=6.5in,height=2.7in]{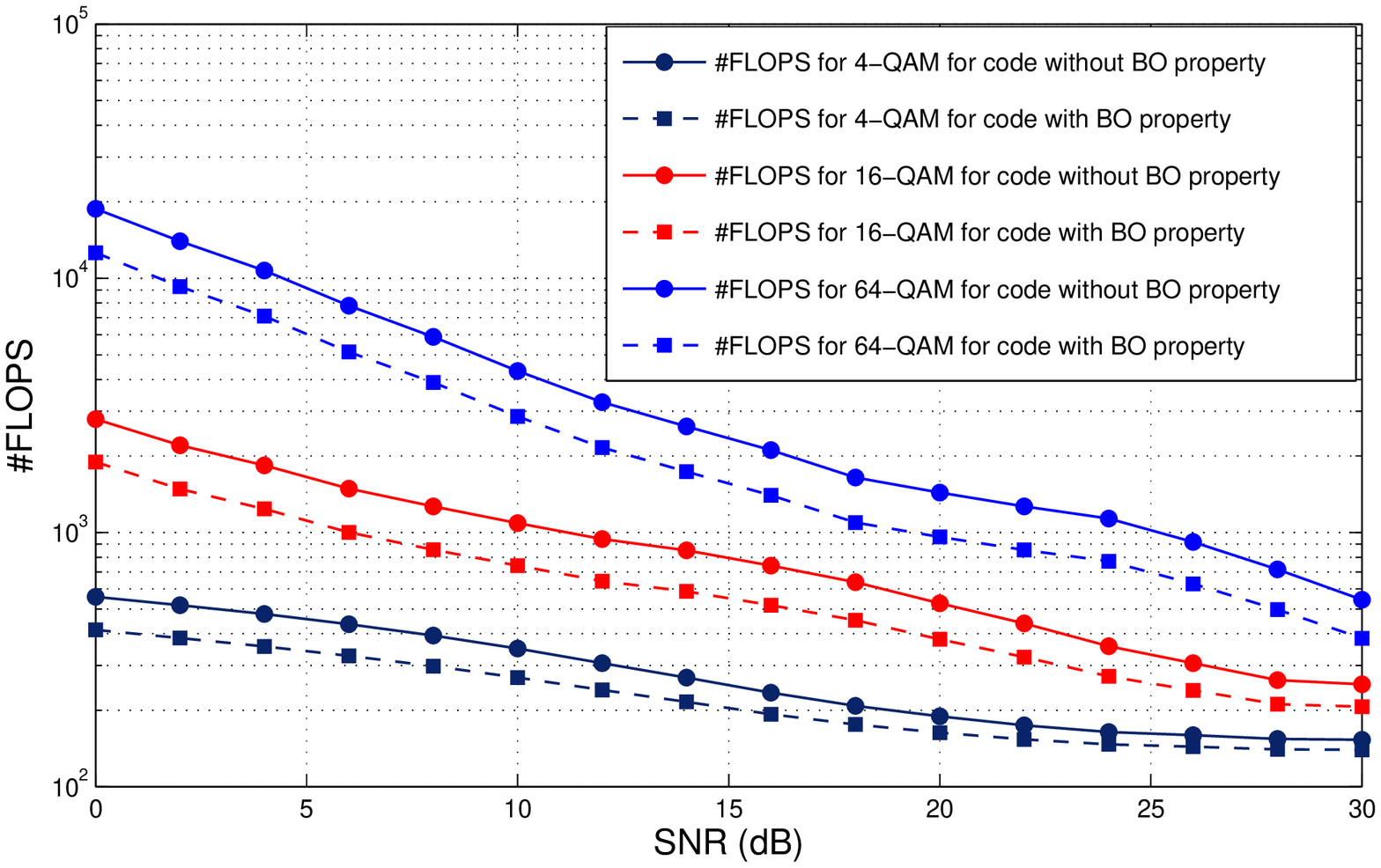}
\caption{The number of FLOPS required for decoding for a BOSTBC with parameters $\left( 2,4,1\right) $} 
\label{fig:stbc_241_flops}
\end{figure}

\begin{figure}
\centering
\includegraphics[width=6.5in,height=2.7in]{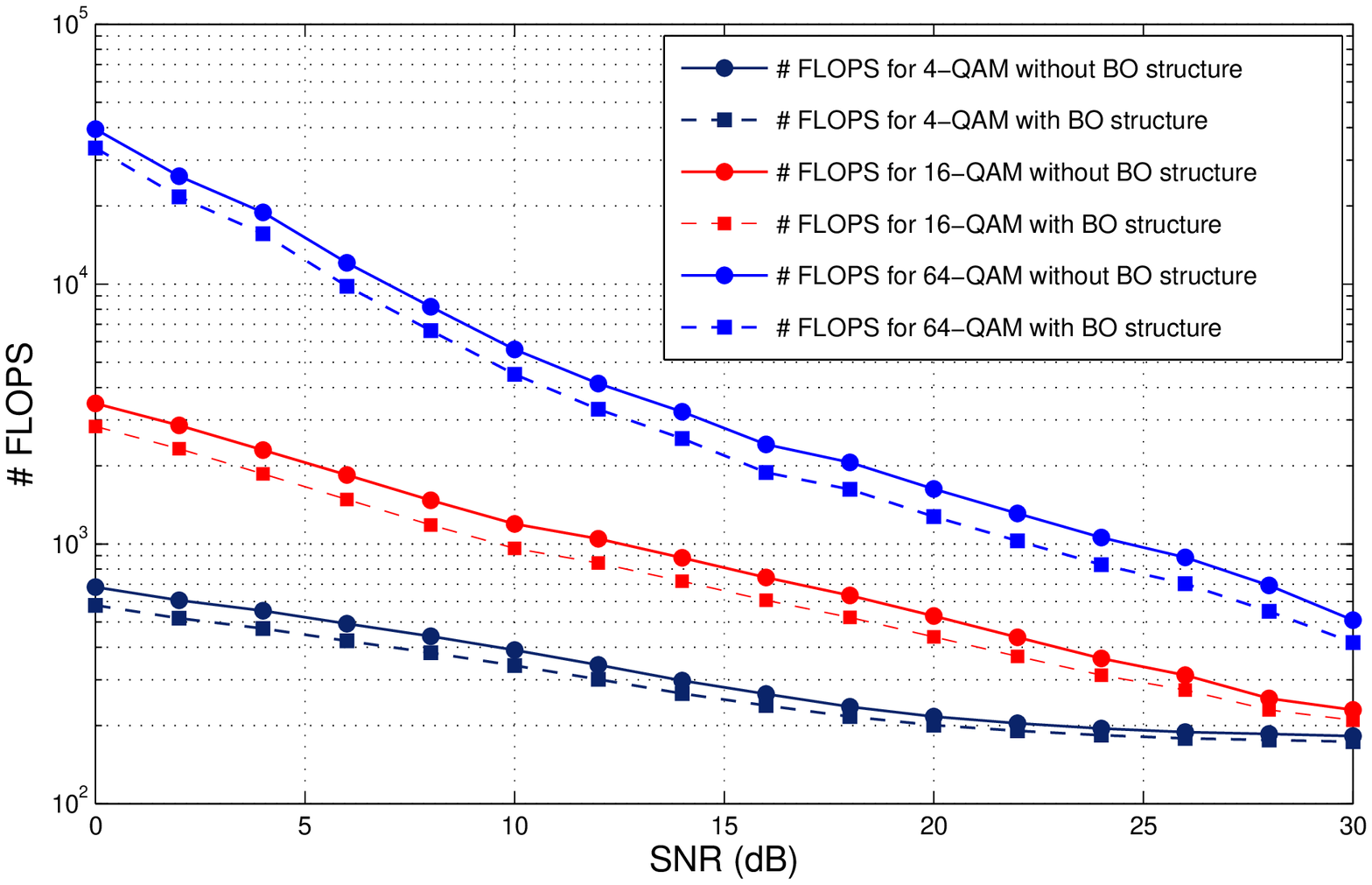}
\caption{The number of FLOPS required for decoding for a BOSTBC with parameters $\left( 2,2,2\right) $} 
\label{fig:stbc_222_flops}
\end{figure}

\begin{figure}
\centering
\includegraphics[width=6.5in,height=2.7in]{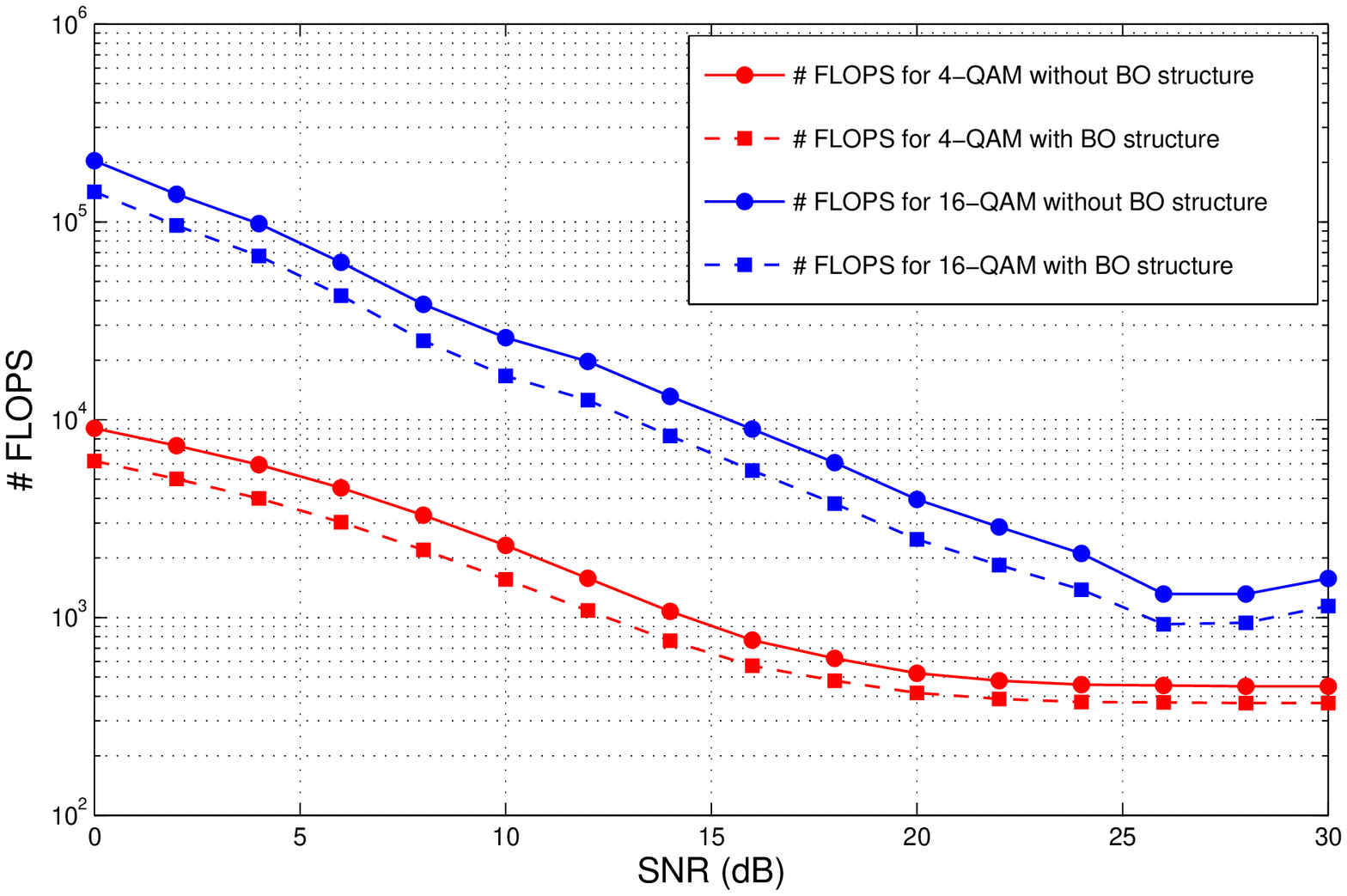}
\caption{The number of FLOPS required for decoding for a BOSTBC with parameters $\left( 2,4,2\right) $} 
\label{fig:stbc_242_flops}
\end{figure}

\subsection{Comparison with the QRDM decoder approach}
\label{compare_ml_qrdm}
The primary difference between the depth first and the breadth first (QRDM) approach is the variation of the EMRR with respect to SNR. As seen in the figures from section \ref{ml_simulations}, the effect of the block orthogonal property reduces as the SNR increases in the depth first sphere decoder. This is owing to the Schnorr-Euchner enumeration and pruning of branches. As the SNR increases, the decoder needs to visit fewer number of nodes in order to find the ML solution and hence the EMRR also tends to 1. However, in the case of a breadth first search algorithm, all the nodes need to be visited in order to arrive to a solution. Hence the EMRR is independent of the SNR in the breadth first search case. To reduce the number of nodes visited, only $M_{c}$ paths are selected in the QRDM algorithm to reduce complexity. The value of $ M_{c}$ chosen needs to be varied with SNR in order to get near ML performance.

\section{Conclusion}
\label{sec7}

In this paper we have studied the block orthogonal property of STBCs. We have shown that this property depends upon the ordering of weight matrices. We have also provided proofs of various existing codes exhibiting the block orthogonal property. A method of exploiting the block orthogonal structure of the STBCs to reduce the sphere decoding complexity was also given with bounds on the maximum possible reduction. 



\newpage

\appendices

\section{Proof of Lemma \ref{bostc_lemma3}}
\label{proof_bostc_lemma3}
Following the system model in Section \ref{sec2}, we have the equivalent channel matrix $ \textbf{H}_{eq} \in \mathbb{R}^{2n_{r}n_{t} \times 2l}$ as $ \textbf{H}_{eq} = \left[ \textbf{H}_{1} ~ \textbf{H}_{2}\right] ~=~ \left[ \textbf{h}_{1} ~...~ \textbf{h}_{l} ~ \textbf{h}_{l+1} ~...~ \textbf{h}_{2l}\right] $. We know from Theorem 2 of \cite{PaR} that, if any two weight matrices $ \textbf{A}_{i}$ and $ \textbf{A}_{j}$ are Hurwitz-Radon orthogonal, then the $i$-th and the $j$-th columns of the $ \textbf{H}_{eq}$ matrix are orthogonal. Due to the conditions on the weight matrices, we have that $ \textbf{H}_{1}^{T} \textbf{H}_{1}$ and $ \textbf{H}_{2}^{T}\textbf{H}_{2}$ are block diagonal with $k$ blocks, each of size $ \gamma \times \gamma$. 
Under $ \textbf{Q} \textbf{R}$ decomposition, $ \textbf{H}_{eq} = \textbf{Q} \textbf{R}$ with $ \textbf{Q} = \left[ \textbf{Q}_{1} ~ \textbf{Q}_{2} \right] $ with $ \textbf{Q}_{1}, \textbf{Q}_{2} \in \mathbb{R}^{2n_{r}n_{t} \times l}$ and $ \textbf{R} = \left[\begin{array}{cc}
\textbf{R}_{1} & \textbf{E}\\
\textbf{0} & \textbf{R}_{2}\\
\end{array}\right]  $ as mentioned. 
It can be seen from Lemma 2 of \cite{JiR} that the matrix $ \textbf{R}_{1}$ is block diagonal with $k$ blocks, each of size $ \gamma \times \gamma$. We can now write,
\begin{equation*}
\textbf{H}_{2} = \textbf{Q}_{1} \textbf{E} + \textbf{Q}_{2} \textbf{R}_{2},
\end{equation*}
\begin{equation*}
\left( \textbf{H}_{2} - \textbf{Q}_{1} \textbf{E}\right)^{T} \left( \textbf{H}_{2} - \textbf{Q}_{1} \textbf{E}\right) = \textbf{R}_{2}^{T} \textbf{Q}_{2}^{T} \textbf{Q}_{2} \textbf{R}_{2}.
\end{equation*}
Simplifying, 
\begin{equation*}
\textbf{H}_{2}^{T} \textbf{H}_{2} - \textbf{E}^{T} \textbf{E} = \textbf{R}_{2}^{T} \textbf{R}_{2}.
\end{equation*}
Now, if $ \textbf{E}^{T} \textbf{E}$ is block diagonal with $k$ blocks of size $ \gamma \times \gamma$ each $ \Rightarrow \textbf{R}_{2}^{T} \textbf{R}_{2}$ is block diagonal with $k$ blocks of size $ \gamma \times \gamma$ each. Since $ \textbf{R}_{2}$ is upper triangular and full rank, this means that $ \textbf{R}_{2}$ is block diagonal with $k$ blocks of size $ \gamma \times \gamma$ each.

\section{Proof of Lemma \ref{bostc_lemma4}}
\label{proof_bostc_lemma4}
Following the system model in Section \ref{sec2}, we have the equivalent channel matrix $ \textbf{H}_{eq} \in \mathbb{R}^{2n_{r}n_{t} \times L+l}$ as $ \textbf{H}_{eq} = \left[ \textbf{H}_{1} ~ \textbf{H}_{2}\right] ~=~ \left[ \textbf{h}_{1} ~...~ \textbf{h}_{L} ~ \textbf{h}_{L+1} ~...~ \textbf{h}_{L+l}\right] $. We know from Theorem 2 of \cite{PaR} that, if any two weight matrices $ \textbf{B}_{i}$ and $ \textbf{B}_{j}$ are Hurwitz-Radon orthogonal, then the $i$-th and the $j$-th columns of the $ \textbf{H}_{eq}$ matrix are orthogonal. Due to the conditions on the weight matrices, we have that $ \textbf{H}_{2}^{T}\textbf{H}_{2}$ is block diagonal with $k$ blocks, each of size $ \gamma \times \gamma$. 
Under $ \textbf{Q} \textbf{R}$ decomposition, $ \textbf{H}_{eq} = \textbf{Q} \textbf{R}$ with $ \textbf{Q} = \left[ \textbf{Q}_{1} ~ \textbf{Q}_{2} \right] $ with $ \textbf{Q}_{1} \in  \mathbb{R}^{2n_{r}n_{t} \times L}$ and $\textbf{Q}_{2} \in \mathbb{R}^{2n_{r}n_{t} \times l}$ and $ \textbf{R} = \left[\begin{array}{cc}
\textbf{R}_{1} & \textbf{E}\\
\textbf{0} & \textbf{R}_{2}\\
\end{array}\right]  $ as mentioned. 
We can now write,
\begin{equation*}
\textbf{H}_{2} = \textbf{Q}_{1} \textbf{E} + \textbf{Q}_{2} \textbf{R}_{2},
\end{equation*}
\begin{equation*}
\left( \textbf{H}_{2} - \textbf{Q}_{1} \textbf{E}\right)^{T} \left( \textbf{H}_{2} - \textbf{Q}_{1} \textbf{E}\right) = \textbf{R}_{2}^{T} \textbf{Q}_{2}^{T} \textbf{Q}_{2} \textbf{R}_{2}.
\end{equation*}
Simplifying, 
\begin{equation*}
\textbf{H}_{2}^{T} \textbf{H}_{2} - \textbf{E}^{T} \textbf{E} = \textbf{R}_{2}^{T} \textbf{R}_{2}.
\end{equation*}
Now, if $ \textbf{E}^{T} \textbf{E}$ is block diagonal with $k$ blocks of size $ \gamma \times \gamma$ each $ \Rightarrow \textbf{R}_{2}^{T} \textbf{R}_{2}$ is block diagonal with $k$ blocks of size $ \gamma \times \gamma$ each. Since $ \textbf{R}_{2}$ is upper triangular and full rank, this means that $ \textbf{R}_{2}$ is block diagonal with $k$ blocks of size $ \gamma \times \gamma$ each.

\section{Structure of the $\textbf{R}$ matrix obtained from Construction I}
\label{app_r_mat_struct_const_1}
According to construction I, the structure of the STBC is
\begin{equation*}
\textbf{X} = \textbf{X}_{1}\left( s_{1}, s_{2}, ..., s_{4 \lambda} \right) + \textbf{MX}_{2} \left( s_{4 \lambda + 1}, s_{4 \lambda + 2}, ..., s_{8 \lambda}\right) ,
\end{equation*}
where $\textbf{X}_{1}$ is a rate-1 four group decodable STBCs obtained from CUWDs as described in Section \ref{cuwds}.

Let the $\textbf{R}$ matrix for this code have the following structure:
\begin{equation*}
\textbf{R} = \left[\begin{array}{cc}
\textbf{R}_{1} & \textbf{E}\\
\textbf{0} & \textbf{R}_{2}\\
\end{array}\right],
\end{equation*}
where $ \textbf{R}_{1}$, $ \textbf{E}$ and $ \textbf{R}_{2}$ are $ 4 \lambda \times 4 \lambda$ matrices. 

\subsection{Structure of $ \textbf{R}_{1}$}
\label{r1_struct_const_1}
From \cite{JiR}, it can be easily seen that $ \textbf{Y}_{1}$ has a block diagonal structure with four blocks, and each block of the size $ \lambda \times \lambda$. 
\begin{equation*}
\textbf{R}_{1} = \left[\begin{array}{cccc}
\textbf{R}_{11} & \textbf{0} & \textbf{0} & \textbf{0}\\
\textbf{0} & \textbf{R}_{12} & \textbf{0} & \textbf{0}\\
\textbf{0} & \textbf{0} & \textbf{R}_{13} & \textbf{0}\\
\textbf{0} & \textbf{0} & \textbf{0} & \textbf{R}_{14}\\
\end{array}\right],
\end{equation*}
where $ \textbf{R}_{1i}$, $i=1,...4$ is a $ \lambda \times \lambda$ given by \eqref{r1i_eq}. 
\begin{figure*}
\begin{equation}
\label{r1i_eq}
\textbf{R}_{1i} = \left[\begin{array}{cccc}
\parallel \textbf{r}_{4\left( i-1\right) \lambda + 1} \parallel & \left\langle \textbf{q}_{4\left( i-1\right) \lambda + 1}, \textbf{h}_{4\left( i-1\right) \lambda + 2}\right\rangle & \cdots & \left\langle \textbf{q}_{4\left( i-1\right) \lambda + 1}, \textbf{h}_{4\left( i-1\right) \lambda + \lambda}\right\rangle\\
0 & \parallel \textbf{r}_{4\left( i-1\right) \lambda + 2} \parallel & \cdots & \left\langle \textbf{q}_{4\left( i-1\right) \lambda + 2}, \textbf{h}_{4\left( i-1\right) \lambda + \lambda}\right\rangle\\
\vdots & \vdots & \ddots & \vdots\\
0 & 0 & \cdots & \parallel \textbf{r}_{4\left( i-1\right) \lambda + \lambda} \parallel\\ 
\end{array}\right],
\end{equation}
\hrule
\end{figure*}

\begin{proposition}
\label{r1_struct_prop_const_1}
The non-zero blocks of the matrix $ \textbf{R}_{1}$ are equal i.e., $ \textbf{R}_{11} = \textbf{R}_{1i}$, for $i=2,3,4$. 
\end{proposition}
\begin{proof}
It is sufficient for us to prove that 
\begin{equation}
\label{r1_struct_r}
\parallel \textbf{r}_{j} \parallel = \parallel \textbf{r}_{4\left( i-1\right) \lambda + j} \parallel
\end{equation}
and
\begin{equation}
\label{r1_struct_qh}
\left\langle \textbf{q}_{j}, \textbf{h}_{k}\right\rangle = \left\langle \textbf{q}_{4\left( i-1\right) \lambda + j}, \textbf{h}_{4\left( i-1\right) \lambda + k}\right\rangle,
\end{equation}
for $i=2,3,4$, $j=1,..., \lambda - 1$ and $k = j+1, ..., \lambda$. 

The proof is by induction. We first consider the case of $j=1$. 
We also recall \cite{PaR} that 
\begin{equation*}
\left\langle \textbf{h}_{k}, \textbf{h}_{j}\right\rangle  = \frac{1}{2} tr\left( \check{\textbf{H}} \check{\textbf{A}}_{k} \check{ \textbf{A}}_{j}^{T} \check{ \textbf{H}}^{T}\right).
\end{equation*}
Now, for \eqref{r1_struct_r} we have, 
\begin{align*}
\parallel \textbf{r}_{1} \parallel ^{2} &= \left\langle \textbf{h}_{1}, \textbf{h}_{1}\right\rangle\\
 & = \frac{1}{2} tr\left( \check{\textbf{H}} \check{\textbf{A}}_{1} \check{ \textbf{A}}_{1}^{T} \check{ \textbf{H}}^{T}\right) \\
 &  = \frac{1}{2} tr\left( \check{\textbf{H}} \check{\textbf{A}}_{4\left( i-1\right) \lambda + 1} \check{ \textbf{A}}_{4\left( i-1\right) \lambda + 1}^{T} \check{ \textbf{H}}^{T}\right) \\
&  = \parallel \textbf{r}_{4\left( i-1\right) \lambda + 1} \parallel ^{2},
\end{align*}
since $\textbf{r}_{4\left( i-1\right) \lambda + 1} = \textbf{h}_{4\left( i-1\right) \lambda + 1}$ and $\check{\textbf{A}}_{k} \check{ \textbf{A}}_{k}^{T} = \check{\textbf{I}}$ for $k = 1, ..., 4 \lambda$. 
For \eqref{r1_struct_qh} we have,
{
\begin{align*}
\left\langle \textbf{q}_{1}, \textbf{h}_{k}\right\rangle &= \frac{1}{\parallel \textbf{r}_{1} \parallel} \left\langle \textbf{h}_{1}, \textbf{h}_{k}\right\rangle\\
&= \frac{tr\left( \check{\textbf{H}} \check{\textbf{A}}_{1} \check{ \textbf{A}}_{k}^{T} \check{ \textbf{H}}^{T}\right)}{2\parallel \textbf{r}_{1} \parallel} \\ 
&= \frac{tr\left( \check{\textbf{H}} \check{\textbf{A}}_{4\left( i-1\right) \lambda + 1} \check{ \textbf{A}}_{4\left( i-1\right) \lambda + 1}^{T} \check{ \textbf{A}}_{k}^{T} \check{ \textbf{H}}^{T}\right)}{2\parallel \textbf{r}_{4\left( i-1\right) \lambda + 1} \parallel}  \\
&= \frac{tr\left( \check{\textbf{H}} \check{\textbf{A}}_{4\left( i-1\right) \lambda + 1} \check{ \textbf{A}}_{4\left( i-1\right) \lambda + k}^{T} \check{ \textbf{H}}^{T}\right)}{2\parallel \textbf{r}_{4\left( i-1\right) \lambda + 1} \parallel}  \\
&= \left\langle \textbf{q}_{4\left( i-1\right) \lambda + 1}, \textbf{h}_{4\left( i-1\right) \lambda + k}\right\rangle,
\end{align*}
}
since $\textbf{A}_{k} \textbf{A}_{4\left( i-1\right) \lambda + 1} = \textbf{A}_{4\left( i-1\right) \lambda + k}$. 
Now we prove equations \eqref{r1_struct_r} and \eqref{r1_struct_qh} for arbitrary $j$. We prove this by induction. Let the equations hold true for all $l < j$. 
We now have for equation \eqref{r1_struct_r},
\begin{align*}
\parallel \textbf{r}_{j} \parallel ^{2} &=  \left\langle \textbf{r}_{j}, \textbf{r}_{j}\right\rangle\\
&=  \left\langle \textbf{h}_{j} - \sum_{l=1}^{j-1}\left\langle \textbf{q}_{l} , \textbf{h}_{j}\right\rangle \textbf{q}_{l} ~,~ \textbf{h}_{j} - \sum_{k=1}^{j-1}\left\langle \textbf{q}_{k} ,  \textbf{h}_{j}\right\rangle \textbf{q}_{k} \right\rangle\\
&=  \left\langle \textbf{h}_{j}, \textbf{h}_{j}\right\rangle - 2 \sum_{l=1}^{j-1}\left\langle \textbf{q}_{l} , \textbf{h}_{j}\right\rangle ^{2} \\
&\quad + \sum_{k=1}^{j-1} \sum_{l=1}^{j-1}\left\langle \textbf{q}_{l} , \textbf{h}_{j}\right\rangle \left\langle \textbf{q}_{k} , \textbf{h}_{j}\right\rangle \left\langle \textbf{q}_{l} , \textbf{q}_{k}\right\rangle\\
&=  \frac{1}{2} tr\left( \check{\textbf{H}} \check{\textbf{A}}_{j} \check{ \textbf{A}}_{j}^{T} \check{ \textbf{H}}^{T}\right)  - 2 \sum_{l=1}^{j-1}\left\langle \textbf{q}_{l} , \textbf{h}_{j}\right\rangle ^{2} \\
&\quad + \sum_{k=1}^{j-1} \sum_{l=1}^{j-1}\left\langle \textbf{q}_{l} , \textbf{h}_{j}\right\rangle \left\langle \textbf{q}_{k} , \textbf{h}_{j}\right\rangle \left\langle \textbf{q}_{l} , \textbf{q}_{k}\right\rangle\\
&=  \frac{1}{2} tr\left( \check{\textbf{H}} \check{\textbf{A}}_{4\left( i-1\right) \lambda + j} \check{ \textbf{A}}_{4\left( i-1\right) \lambda + j}^{T} \check{ \textbf{H}}^{T}\right)  \\
&\quad - 2 \sum_{l=1}^{j-1}\left\langle \textbf{q}_{4\left( i-1\right) \lambda + l} , \textbf{h}_{4\left( i-1\right) \lambda + j}\right\rangle ^{2}\\ 
& \quad + \sum_{k=1}^{j-1} \sum_{l=1}^{j-1}\left\langle \textbf{q}_{4\left( i-1\right) \lambda + l} , \textbf{h}_{4\left( i-1\right) \lambda + j}\right\rangle .\\
&\quad \left\langle \textbf{q}_{4\left( i-1\right) \lambda + k} , \textbf{h}_{4\left( i-1\right) \lambda + j}\right\rangle \left\langle \textbf{q}_{4\left( i-1\right) \lambda + l} , \textbf{q}_{4\left( i-1\right) \lambda + k}\right\rangle\\
&= ~\parallel \textbf{r}_{4\left( i-1\right) \lambda + j} \parallel ^{2} ,
\end{align*}
which follows from the induction hypothesis and the fact that $\check{\textbf{A}}_{j} \check{ \textbf{A}}_{j}^{T} = \check{\textbf{I}}$ for $j = 1, ..., 4 \lambda$ .
For equation \eqref{r1_struct_qh}, 
\begin{align*}
\left\langle \textbf{q}_{j}, \textbf{h}_{k}\right\rangle &= \frac{1}{\parallel \textbf{r}_{j} \parallel} \left\langle \textbf{h}_{j} - \sum_{l=1}^{j-1} \left\langle \textbf{q}_{l}, \textbf{h}_{j}\right\rangle \textbf{q}_{l}~,~ \textbf{h}_{k}\right\rangle\\
&= \frac{1}{\parallel \textbf{r}_{j} \parallel} \left[ \left\langle \textbf{h}_{j}, \textbf{h}_{k}\right\rangle - \sum_{l=1}^{j-1} \left\langle \textbf{q}_{l}, \textbf{h}_{j}\right\rangle \left\langle \textbf{q}_{l}, \textbf{h}_{k}\right\rangle\right] \\
&= \frac{1}{2\parallel \textbf{r}_{j} \parallel} \left[ tr\left( \check{\textbf{H}} \check{\textbf{A}}_{j} \check{ \textbf{A}}_{k}^{T} \check{ \textbf{H}}^{T}\right) - \sum_{l=1}^{j-1} \left\langle \textbf{q}_{l}, \textbf{h}_{j}\right\rangle \left\langle \textbf{q}_{l}, \textbf{h}_{k}\right\rangle\right] \\
&= \frac{tr\left( \check{\textbf{H}} \check{\textbf{A}}_{j} \check{\textbf{A}}_{4\left( i-1\right) \lambda + 1} \check{ \textbf{A}}_{4\left( i-1\right) \lambda + 1}^{T} \check{ \textbf{A}}_{k}^{T} \check{ \textbf{H}}^{T}\right)}{2\parallel \textbf{r}_{4\left( i-1\right) \lambda + j} \parallel}  \\
&\quad - \frac{1}{2\parallel \textbf{r}_{4\left( i-1\right) \lambda + j} \parallel}\sum_{l=1}^{j-1} \left\langle \textbf{q}_{4\left( i-1\right) \lambda + l}, \textbf{h}_{4\left( i-1\right) \lambda + j}\right\rangle.\\
&\qquad \left\langle \textbf{q}_{4\left( i-1\right) \lambda + l}, \textbf{h}_{4\left( i-1\right) \lambda + k}\right\rangle
\end{align*}
\begin{align*}
&= \frac{tr\left( \check{\textbf{H}} \check{\textbf{A}}_{4\left( i-1\right) \lambda + j} \check{ \textbf{A}}_{4\left( i-1\right) \lambda + k}^{T} \check{ \textbf{H}}^{T}\right) }{2\parallel \textbf{r}_{4\left( i-1\right) \lambda + j} \parallel} \\
&\quad - \frac{1}{2\parallel \textbf{r}_{4\left( i-1\right) \lambda + j} \parallel}\sum_{l=1}^{j-1} \left\langle \textbf{q}_{4\left( i-1\right) \lambda + l}, \textbf{h}_{4\left( i-1\right) \lambda + j}\right\rangle.\\
&\qquad \left\langle \textbf{q}_{4\left( i-1\right) \lambda + l}, \textbf{h}_{4\left( i-1\right) \lambda + k}\right\rangle\\
&= \left\langle \textbf{q}_{4\left( i-1\right) \lambda + j}, \textbf{h}_{4\left( i-1\right) \lambda + k}\right\rangle .
\end{align*}

\end{proof}

\subsection{Structure of $ \textbf{E}$}
\label{e_struct_const_1}

The matrix $ \textbf{E}$ is key for the block orthogonality property of the STBC in question. It is required to be para-unitary for achieving this property. The structure of the matrix $ \textbf{E}$ for Construction I is described in the following proposition.

\begin{proposition}
\label{e_struct_prop_const_1}
The matrix $ \textbf{E}$ is of the form
\begin{equation}
\label{e_struct_eq}
\textbf{E} = \left[\begin{array}{cccc}
\textbf{E}_{1} & -\textbf{E}_{2} & -\textbf{E}_{3} & -\textbf{E}_{4}\\
\textbf{E}_{2} & \textbf{E}_{1} & -\textbf{E}_{4} \textbf{P} & \textbf{E}_{3} \textbf{P}\\
\textbf{E}_{3} & \textbf{E}_{4} \textbf{P} & \textbf{E}_{1} & -\textbf{E}_{2} \textbf{P}\\
\textbf{E}_{4} & -\textbf{E}_{3} \textbf{P} & \textbf{E}_{2} \textbf{P} & \textbf{E}_{1}\\
\end{array}\right],
\end{equation}
where $ \textbf{E}_{i}$, $i = 1,...,4$ are $ \lambda \times \lambda$ matrices and $ \textbf{P}$ is a $ \lambda \times \lambda$ permutation matrix given by
\begin{equation*}
\textbf{P} = \left[\begin{array}{ccccc}
0 & 0 & \cdots & 0 & 1\\
0 & 0 & \cdots & 1 & 0\\
\vdots & \vdots & \ddots & \vdots & \vdots\\
0 & 1 & \cdots & 0 & 0\\
1 & 0 & \cdots & 0 & 0\\
\end{array}\right].
\end{equation*}
\end{proposition} 

\begin{proof}
Let us represent the matrix $ \textbf{E}$ using $ \lambda \times \lambda$ blocks as:
\begin{equation*}
\textbf{E} = \left[\begin{array}{cccc}
\textbf{E}_{11} & \textbf{E}_{12} & \textbf{E}_{13} & \textbf{E}_{14}\\
\textbf{E}_{21} & \textbf{E}_{22} & \textbf{E}_{23} & \textbf{E}_{24}\\
\textbf{E}_{31} & \textbf{E}_{32} & \textbf{E}_{33} & \textbf{E}_{34}\\
\textbf{E}_{41} & \textbf{E}_{42} & \textbf{E}_{43} & \textbf{E}_{44}\\
\end{array}\right],
\end{equation*}

We first prove that $ \textbf{E}_{11} = \textbf{E}_{ii}$ for $i=2,3,4$. The proof is by induction on the rows of the matrix $ \textbf{E}_{11}$. The first row entries of the matrix $ \textbf{E}_{11}$ are given by
\begin{equation*}
\textbf{E}_{11} \left( 1, k\right) = \left\langle \textbf{q}_{1}, \textbf{h}_{4 \lambda + k} \right\rangle 
\end{equation*}
and for the matrix $ \textbf{E}_{ii}$ are given by
\begin{align*}
\textbf{E}_{ii} \left( 1, k\right) &= \left\langle \textbf{q}_{4\left( i-1\right) \lambda + 1}, \textbf{h}_{4 \lambda + 4\left( i-1\right) \lambda + k} \right\rangle\\ 
&= \frac{\left\langle \textbf{h}_{4\left( i-1\right) \lambda + 1}, \textbf{h}_{4 \lambda + 4\left( i-1\right) \lambda + k} \right\rangle}{\parallel \textbf{r}_{4\left( i-1\right) \lambda + 1} \parallel} .
\end{align*}
Due to the construction of the STBC, we have $ \textbf{A}_{4 \lambda + l} =  \textbf{M}\textbf{A}_{l}$, for $l = 1, ..., 4 \lambda$. Using this, we get
\begin{align*}
\textbf{E}_{ii} \left( 1, k\right) &= \frac{tr\left( \check{\textbf{H}} \check{\textbf{A}}_{4\left( i-1\right) \lambda + 1} \check{ \textbf{A}}_{4\left( i-1\right) \lambda + k}^{T} \check{ \textbf{M}}^{T} \check{ \textbf{H}}^{T}\right)}{2\parallel \textbf{r}_{4\left( i-1\right) \lambda + 1} \parallel}\\
&= \frac{tr\left( \check{\textbf{H}} \check{\textbf{A}}_{4\left( i-1\right) \lambda + 1} \check{ \textbf{A}}_{4\left( i-1\right) \lambda + 1}^{T} \check{ \textbf{A}}_{k}^{T} \check{ \textbf{M}}^{T} \check{ \textbf{H}}^{T}\right)}{2\parallel \textbf{r}_{4\left( i-1\right) \lambda + 1} \parallel}\\
&= \frac{tr\left( \check{\textbf{H}} \check{\textbf{A}}_{1} \check{ \textbf{A}}_{k}^{T} \check{ \textbf{M}}^{T} \check{ \textbf{H}}^{T}\right)}{2\parallel \textbf{r}_{1} \parallel}\\
&= \left\langle \textbf{q}_{1}, \textbf{h}_{4 \lambda + k} \right\rangle\\
&= \textbf{E}_{11} \left( 1, k\right) .
\end{align*}
Now, let us assume that row $m$ of $ \textbf{E}_{ii}$ is equal to the row $m$ of $ \textbf{E}_{11}$ for all $m < j$. The $j$-th row of $ \textbf{E}_{11}$ is given by
\begin{equation*}
\textbf{E}_{11} \left( j, k\right) = \left\langle \textbf{q}_{j}, \textbf{h}_{4 \lambda + k} \right\rangle ,
\end{equation*}
and the $j$-th row of $ \textbf{E}_{ii}$ is given by
\begin{align*}
\textbf{E}_{ii} \left( j, k\right) &= \left\langle \textbf{q}_{4\left( i-1\right) \lambda + j}~,~ \textbf{h}_{4 \lambda + 4\left( i-1\right) \lambda + k} \right\rangle \\
&= \frac{1}{\parallel \textbf{r}_{4\left( i-1\right) \lambda + j} \parallel} \left\langle \textbf{h}_{4\left( i-1\right) \lambda + j}\right.  \\
&\quad \left. - \sum_{m=1}^{j-1} \left\langle \textbf{q}_{4\left( i-1\right) \lambda + m}, \textbf{h}_{4\left( i-1\right) \lambda + j} \right\rangle \textbf{q}_{4\left( i-1\right) \lambda + m}, \right.\\
&\qquad \left. \textbf{h}_{4 \lambda + 4\left( i-1\right) \lambda + k} \right\rangle \\
&= \frac{1}{\parallel \textbf{r}_{4\left( i-1\right) \lambda + j} \parallel} \left\langle \textbf{h}_{4\left( i-1\right) \lambda + j}, \textbf{h}_{4 \lambda + 4\left( i-1\right) \lambda + k} \right\rangle\\
&\quad  - \frac{1}{\parallel \textbf{r}_{4\left( i-1\right) \lambda + j} \parallel}\sum_{m=1}^{j-1} \left\langle \textbf{q}_{4\left( i-1\right) \lambda + m}, \textbf{h}_{4\left( i-1\right) \lambda + j} \right\rangle.\\
&\qquad \left\langle  \textbf{q}_{4\left( i-1\right) \lambda + m} , \textbf{h}_{4 \lambda + 4\left( i-1\right) \lambda + k} \right\rangle \\
&= \frac{tr\left( \check{\textbf{H}} \check{\textbf{A}}_{4\left( i-1\right) \lambda + j} \check{ \textbf{A}}_{4\left( i-1\right) \lambda + k}^{T} \check{ \textbf{M}}^{T} \check{ \textbf{H}}^{T}\right)}{2\parallel \textbf{r}_{4\left( i-1\right) \lambda + j} \parallel}  \\
&\quad - \frac{1}{2\parallel \textbf{r}_{4\left( i-1\right) \lambda + j} \parallel}\sum_{m=1}^{j-1} \left\langle \textbf{q}_{m}, \textbf{h}_{j} \right\rangle \left\langle  \textbf{q}_{m} , \textbf{h}_{4 \lambda + k} \right\rangle
\end{align*}
\begin{align*}
&= \frac{tr\left( \check{\textbf{H}} \check{\textbf{A}}_{j} \check{\textbf{A}}_{4\left( i-1\right) \lambda + 1} \check{ \textbf{A}}_{4\left( i-1\right) \lambda + 1}^{T} \check{\textbf{A}}_{k}^{T} \check{ \textbf{M}}^{T} \check{ \textbf{H}}^{T}\right)}{2\parallel \textbf{r}_{j} \parallel}  \\
&\quad - \frac{1}{2\parallel \textbf{r}_{j} \parallel} \sum_{m=1}^{j-1} \left\langle \textbf{q}_{m}, \textbf{h}_{j} \right\rangle \left\langle  \textbf{q}_{m} , \textbf{h}_{4 \lambda + k} \right\rangle \\
&= \frac{1}{2\parallel \textbf{r}_{j} \parallel} \left\langle \textbf{h}_{j}, \textbf{h}_{4 \lambda + k} \right\rangle  \\
&\quad- \frac{1}{2\parallel \textbf{r}_{j} \parallel}\sum_{m=1}^{j-1} \left\langle \textbf{q}_{m}, \textbf{h}_{j} \right\rangle \left\langle  \textbf{q}_{m} , \textbf{h}_{4 \lambda + k} \right\rangle \\
& = \textbf{E}_{11} \left( j, k\right).
\end{align*}

We now prove that $ \textbf{E}_{2} =  \textbf{E}_{21} = -\textbf{E}_{12}= \textbf{E}_{43} \textbf{P} = -\textbf{E}_{34} \textbf{P} $. The proofs for the matrices $ \textbf{E}_{3}$ and $ \textbf{E}_{4}$ are very similar. First step is to prove that $ \textbf{E}_{21} = - \textbf{E}_{12}$. 
The proof is by induction on the rows of the matrix $ \textbf{E}_{21}$. The first row entries of the matrix $ \textbf{E}_{21}$ are given by
\begin{align*}
\textbf{E}_{21} \left( 1, k\right) &= \left\langle \textbf{q}_{ \lambda + 1}, \textbf{h}_{4 \lambda + k} \right\rangle \\
& = \frac{1}{2\parallel \textbf{r}_{1} \parallel}tr\left( \check{\textbf{H}} \check{\textbf{A}}_{ \lambda + 1} \check{ \textbf{A}}_{\lambda + k}^{T} \check{ \textbf{M}}^{T} \check{ \textbf{H}}^{T}\right)
\end{align*}
and for the matrix $ \textbf{E}_{12}$ are given by
\begin{align*}
\textbf{E}_{12} \left( 1, k\right) &= \left\langle \textbf{q}_{1}, \textbf{h}_{5 \lambda + k} \right\rangle \\
&= \frac{1}{\parallel \textbf{r}_{1} \parallel}\left\langle \textbf{h}_{1}, \textbf{h}_{5 \lambda + k} \right\rangle .
\end{align*}
Due to the construction of the STBC, we have $ \textbf{A}_{4 \lambda + l} =  \textbf{M}\textbf{A}_{l}$, for $l = 1, ..., 4 \lambda$. Using this, we get
\begin{align*}
\textbf{E}_{12} \left( 1, k\right) &= \frac{1}{2\parallel \textbf{r}_{1} \parallel}tr\left( \check{\textbf{H}} \check{\textbf{A}}_{1} \check{ \textbf{A}}_{\lambda + k}^{T} \check{ \textbf{M}}^{T} \check{ \textbf{H}}^{T}\right)\\
&= \frac{1}{2\parallel \textbf{r}_{1} \parallel}tr\left( \check{\textbf{H}} \check{\textbf{A}}_{1} \check{ \textbf{A}}_{\lambda + 1}^{T} \check{ \textbf{A}}_{k}^{T} \check{ \textbf{M}}^{T} \check{ \textbf{H}}^{T}\right)\\
&= - \frac{1}{2\parallel \textbf{r}_{1} \parallel}tr\left( \check{\textbf{H}} \check{\textbf{A}}_{ \lambda + 1} \check{ \textbf{A}}_{k}^{T} \check{ \textbf{M}}^{T} \check{ \textbf{H}}^{T}\right)\\
&= \left\langle \textbf{q}_{ \lambda + 1}, \textbf{h}_{4 \lambda + k} \right\rangle\\ 
&= \textbf{E}_{21} \left( 1, k\right) .
\end{align*}
Now, let us assume that row $m$ of $ \textbf{E}_{21}$ is equal to the row $m$ of $ \textbf{E}_{12}$ for all $m < j$. The $j$-th row of $ \textbf{E}_{21}$ is given by
\begin{equation*}
\textbf{E}_{21} \left( j, k\right) = \left\langle \textbf{q}_{ \lambda + j}, \textbf{h}_{4 \lambda + k} \right\rangle ,
\end{equation*}
and the $j$-th row of $ \textbf{E}_{12}$ is given by
\begin{align*}
\textbf{E}_{12} \left( j, k\right) &= \left\langle \textbf{q}_{j}, \textbf{h}_{5 \lambda + k} \right\rangle\\ 
&= \frac{1}{\parallel \textbf{r}_{j} \parallel} \left\langle \textbf{h}_{j} - \sum_{m=1}^{j-1} \left\langle \textbf{q}_{m}, \textbf{h}_{j} \right\rangle \textbf{q}_{m} ~,~ \textbf{h}_{5 \lambda + k} \right\rangle \\
&= \frac{\left\langle \textbf{h}_{j}, \textbf{h}_{5 \lambda k} \right\rangle  - \sum_{m=1}^{j-1} \left\langle \textbf{q}_{m}, \textbf{h}_{j} \right\rangle \left\langle  \textbf{q}_{m} , \textbf{h}_{5 \lambda + k} \right\rangle}{\parallel \textbf{r}_{ \lambda + j} \parallel}  \\
&= \frac{1}{2\parallel \textbf{r}_{\lambda + j} \parallel} tr\left( \check{\textbf{H}} \check{\textbf{A}}_{j} \check{ \textbf{A}}_{ \lambda + k}^{T} \check{ \textbf{M}}^{T} \check{ \textbf{H}}^{T}\right)\\
& \quad  - \frac{1}{2\parallel \textbf{r}_{\lambda + j} \parallel}\sum_{m=1}^{j-1} \left\langle \textbf{q}_{m}, \textbf{h}_{j} \right\rangle \left\langle  \textbf{q}_{m} , \textbf{h}_{5 \lambda + k} \right\rangle
\end{align*}
\begin{align*}
&\quad= \frac{1}{2\parallel \textbf{r}_{\lambda + j} \parallel} tr\left( \check{\textbf{H}} \check{\textbf{A}}_{j} \check{ \textbf{A}}_{\lambda + 1}^{T} \check{\textbf{A}}_{k}^{T} \check{ \textbf{M}}^{T} \check{ \textbf{H}}^{T}\right) \\
&\qquad - \frac{1}{2\parallel \textbf{r}_{\lambda + j} \parallel}\sum_{m=1}^{j-1} \left\langle \textbf{q}_{m}, \textbf{h}_{j} \right\rangle \left\langle  \textbf{q}_{m} , \textbf{h}_{5 \lambda + k} \right\rangle \\
&\quad= \frac{1}{2\parallel \textbf{r}_{\lambda + j} \parallel} tr\left( \check{\textbf{H}} \check{\textbf{A}}_{ \lambda + j} \check{ \textbf{A}}_{k}^{T} \check{ \textbf{M}}^{T} \check{ \textbf{H}}^{T}\right) \\ 
&\qquad + \frac{1}{2\parallel \textbf{r}_{\lambda + j} \parallel} \sum_{m=1}^{j-1} \left\langle \textbf{q}_{ \lambda + m}, \textbf{h}_{ \lambda + j} \right\rangle \left\langle  \textbf{q}_{ \lambda + m} , \textbf{h}_{4 \lambda + k} \right\rangle \\
&\quad= -\frac{1}{2\parallel \textbf{r}_{ \lambda + j} \parallel} \left\langle \textbf{h}_{ \lambda + j}, \textbf{h}_{4 \lambda + k} \right\rangle \\ 
&\qquad + \frac{1}{2\parallel \textbf{r}_{ \lambda + j} \parallel}\sum_{m=1}^{j-1} \left\langle \textbf{q}_{ \lambda + m}, \textbf{h}_{ \lambda + j} \right\rangle \left\langle  \textbf{q}_{ \lambda + m} , \textbf{h}_{4 \lambda + k} \right\rangle  \\
&\quad = \textbf{E}_{21} \left( j, k\right).
\end{align*}

We now prove that $ \textbf{E}_{12} = \textbf{E}_{43} \textbf{P}$. The proof is by induction on the rows of the matrix $ \textbf{E}_{12}$. The first row entries of the matrix $ \textbf{E}_{12}$ are given by
\begin{align*}
\textbf{E}_{12} \left( 1, k\right) &= \left\langle \textbf{q}_{1}, \textbf{h}_{5 \lambda + k} \right\rangle \\
&= \frac{1}{\parallel \textbf{r}_{1} \parallel}\left\langle \textbf{h}_{1}, \textbf{h}_{5 \lambda + k} \right\rangle .
\end{align*}
Due to the construction of the STBC, we have $ \textbf{A}_{4 \lambda + l} =  \textbf{M}\textbf{A}_{l}$, for $l = 1, ..., 4 \lambda$. Using this, we get
\begin{align*}
\textbf{E}_{12} \left( 1, k\right) &= \frac{1}{2\parallel \textbf{r}_{1} \parallel}tr\left( \check{\textbf{H}} \check{\textbf{A}}_{1} \check{ \textbf{A}}_{\lambda + k}^{T} \check{ \textbf{M}}^{T} \check{ \textbf{H}}^{T}\right)\\
&= \frac{1}{2\parallel \textbf{r}_{1} \parallel}tr\left( \check{\textbf{H}} \check{ \textbf{A}}_{\lambda + 1}^{T} \check{ \textbf{A}}_{k}^{T} \check{ \textbf{M}}^{T} \check{ \textbf{H}}^{T}\right).
\end{align*}
We need to show that this is equal to $ -\textbf{E}_{43} \left( 1, \lambda - k + 1\right) $. 
\begin{align*}
\textbf{E}_{43} \left( 1, \lambda - k +1 \right) &= \left\langle \textbf{q}_{3 \lambda + 1}, \textbf{h}_{11 \lambda - k + 1} \right\rangle \\
&= \frac{1}{\parallel \textbf{r}_{3 \lambda + 1} \parallel}\left\langle \textbf{h}_{3 \lambda + 1}, \textbf{h}_{11 \lambda - k + 1} \right\rangle \\
&= \frac{tr\left( \check{\textbf{H}} \check{\textbf{A}}_{3 \lambda + 1} \check{ \textbf{A}}_{3 \lambda - k + 1}^{T} \check{ \textbf{M}}^{T} \check{ \textbf{H}}^{T}\right)}{2\parallel \textbf{r}_{1} \parallel}\\
&= \frac{tr\left( \check{\textbf{H}} \check{\textbf{A}}_{3 \lambda + 1} \check{ \textbf{A}}_{2 \lambda + 1}^{T} \check{ \textbf{A}}_{ \lambda - k + 1}^{T}\check{ \textbf{M}}^{T} \check{ \textbf{H}}^{T}\right)}{2\parallel \textbf{r}_{1} \parallel}\\
&= -\frac{tr\left( \check{\textbf{H}} \check{\textbf{A}}_{3 \lambda + 1} \check{ \textbf{A}}_{2 \lambda + 1}^{T} \check{ \textbf{A}}_{ \lambda - k + 1}^{T}\check{ \textbf{M}}^{T} \check{ \textbf{H}}^{T}\right)}{2\parallel \textbf{r}_{1} \parallel}.
\end{align*}
Substituting the values of the weight matrices from \eqref{cuwd_mat} for $ \textbf{A}_{ \lambda + 1}$, $ \textbf{A}_{2 \lambda + 1}$ and $ \textbf{A}_{3 \lambda + 1}$, and simplifying, we see that it is sufficient to show that 
\begin{equation*}
\left( \textbf{I}_{2}^{\otimes a-1} \bigotimes j \sigma_{3}\right) \textbf{A}_{ \lambda - k + 1}^{T} = j \sigma_{3}^{ \otimes a} \textbf{A}_{k},
\end{equation*} 
or equivalently,
\begin{equation*}
\left( \textbf{I}_{2}^{\otimes a-1} \bigotimes j \sigma_{3}\right) \textbf{A}_{ \lambda - k + 1}^{T} \textbf{A}_{k} ^{T} = j \sigma_{3}^{ \otimes a} .
\end{equation*} 
Since $ \lambda - k + 1$ and $k$ are one's complement of each other in the binary representation, we have, 
\begin{equation*}
\textbf{A}_{ \lambda - k + 1}\textbf{A}_{k} = \prod_{i=1}^{a-1} \alpha_{i} = \textbf{A}_{ \lambda} = j \sigma_{3}^{\otimes a-1} \bigotimes \textbf{I}_{2}.
\end{equation*}
Therefore we have,
\begin{align*}
\left( \textbf{I}_{2}^{\otimes a-1} \bigotimes j \sigma_{3}\right) \textbf{A}_{ \lambda} &= \left( \textbf{I}_{2}^{\otimes a-1} \bigotimes j \sigma_{3}\right) \left( j \sigma_{3}^{\otimes a-1} \bigotimes \textbf{I}_{2} \right) \\
&= j \sigma_{3}^{ \otimes a} .
\end{align*} 

The equality for $ \textbf{E}_{3}$ and $ \textbf{E}_{4}$ can be shown similarly.
\end{proof}

\subsection{Structure of $ \textbf{R}_{2}$}
\label{r2_struct_const_1}

\begin{proposition}
\label{r2_struct_prop_const_1}
The matrix $ \textbf{R}_{2}$ is block diagonal with $4$ blocks, each of size $ \lambda \times \lambda$.
\end{proposition}
\begin{proof}
For the matrix $ \textbf{R}_{2}$ to be block diagonal with $4$ blocks, each of size $ \lambda \times \lambda$, we need to satisfy the following conditions
\begin{itemize}
\item The matrices $ \left\lbrace \textbf{M} \textbf{A}_{1}, \textbf{M} \textbf{A}_{2}, ..., \textbf{M} \textbf{A}_{4 \lambda}\right\rbrace $ form a four group decodable STBC with $ \lambda$ variables per group
\item The matrix $ \textbf{E}$ is such that $ \textbf{E}^{T} \textbf{E}$ is block diagonal with $4$ blocks, each of size $ \lambda \times \lambda$.
\end{itemize}

Since the matrices $ \left\lbrace \textbf{A}_{1}, \textbf{A}_{2}, ...,  \textbf{A}_{4 \lambda}\right\rbrace $ form a four group decodable STBC with $ \lambda$ variables per group, it is easily seen that the matrices  $ \left\lbrace \textbf{M} \textbf{A}_{1}, \textbf{M} \textbf{A}_{2}, ..., \textbf{M} \textbf{A}_{4 \lambda}\right\rbrace $ also form a four group decodable STBC with $ \lambda$ variables per group as $ \left( \textbf{M} \textbf{A}_{i} \right) \left( \textbf{M} \textbf{A}_{j} \right) ^{H} + \left( \textbf{M} \textbf{A}_{j} \right) \left( \textbf{M} \textbf{A}_{i} \right) ^{H} = \textbf{M}\left[ \textbf{A}_{i} \textbf{A}_{j}^{H} + \textbf{A}_{j} \textbf{A}_{i}^{H}\right] \textbf{M}^{H} = \textbf{0}$ for $i$ and $j$ in different groups. 

We now introduce some notation before we address the structure of the matrix $ \textbf{E}^{H} \textbf{E}$. Let $m$ be an integer such that $ 1 \leq m \leq \lambda$. We denote by $f\left( m\right)$, the binary representation of $m-1$ using $a-1$ bits. Let $ \oplus$ denote the bitwise XOR operation between any two binary numbers. 

Now, we turn to the structure of the matrix $ \textbf{E}$. From Proposition \ref{e_struct_prop_const_1}, we know the structure of the matrix $ \textbf{E}$. Computing $ \textbf{E}^{T} \textbf{E}$, we see that for it to be block diagonal with $4$ blocks, each of size $ \lambda \times \lambda$, it is sufficient to show that the matrices $ \textbf{E}_{i}^{T} \textbf{E}_{j} $ are symmetric with identical entries on the diagonal for $i,j = 1, ..., 4, ~i \neq j$. The entries of $ \textbf{E}_{i}^{T} \textbf{E}_{j} $ are given by
\begin{equation*}
\textbf{E}_{i}^{T} \textbf{E}_{j}\left( k,l\right) = \sum_{m=1}^{\lambda} \left\langle  \textbf{q}_{m}, \textbf{h}_{4 \lambda + k} \right\rangle  \left\langle  \textbf{q}_{\lambda + m}, \textbf{h}_{4 \lambda + l} \right\rangle.
\end{equation*}
Expanding and simplifying, we get
\begin{equation*}
\textbf{E}_{i}^{T} \textbf{E}_{j}\left( k,l\right) = \sum_{m=1}^{\lambda}\sum_{n=1}^{\lambda} a_{mn} \left\langle  \textbf{h}_{m}, \textbf{h}_{4 \lambda + k} \right\rangle  \left\langle  \textbf{h}_{\lambda + n}, \textbf{h}_{4 \lambda + l} \right\rangle,
\end{equation*}
where $a_{mn} = a_{t},$ $t = f^{-1}\left( f\left( m\right) \oplus f\left( n\right) \right) $ and $a_{t}$ is given by 
\begin{equation*}
a_{t} = \frac{ -\sum_{p=1}^{t-1} a_{p} \left\langle \textbf{q}_{\lambda - t + 1}, \textbf{h}_{\lambda - p + 1}\right\rangle}{ \parallel \textbf{r}_{\lambda - t + 1} \parallel},
\end{equation*}
for $t = 2,3,.., \lambda$ and $a_{1} = \frac{1}{\parallel \textbf{r}_{\lambda} \parallel^{2}}$. 
We now see that for every $m$, there exists a unique $m^{'}$ such that $\left\langle  \textbf{h}_{m}, \textbf{h}_{4 \lambda + k} \right\rangle = \left\langle  \textbf{h}_{m^{'}}, \textbf{h}_{4 \lambda + l}\right\rangle $ as
\begin{align*}
\left\langle  \textbf{h}_{m}, \textbf{h}_{4 \lambda + k} \right\rangle &= tr\left( \check{\textbf{H}} \check{\textbf{A}}_{m} \check{ \textbf{A}}_{\lambda + k}^{T} \check{ \textbf{M}}^{T} \check{ \textbf{H}}^{T}\right)\\
 &= tr\left( \check{\textbf{H}} \left[ \check{\textbf{A}}_{m} \check{ \textbf{A}}_{\lambda + k}^{T} \check{ \textbf{A}}_{\lambda + l}\right]  \check{ \textbf{A}}_{\lambda + l}^{T}\check{ \textbf{M}}^{T} \check{ \textbf{H}}^{T}\right)\\
 &= \left\langle  \textbf{h}_{m^{'}}, \textbf{h}_{4 \lambda + l} \right\rangle,
\end{align*}
where $m^{'} = f^{-1}\left( f\left( m\right) \oplus f\left( k\right) \oplus f\left( l\right) \right)$. Similarly, for every $n$, there exists a unique $n^{'}$ such that $\left\langle  \textbf{h}_{\lambda + n}, \textbf{h}_{4 \lambda + l} \right\rangle = \left\langle  \textbf{h}_{\lambda + n^{'}}, \textbf{h}_{4 \lambda + k}\right\rangle $ where $n^{'} = f^{-1}\left( f\left( n\right) \oplus f\left( k\right) \oplus f\left( l\right) \right)$. We can now write,
\begin{align*}
\textbf{E}_{i}^{T} \textbf{E}_{j}\left( k,l\right) &= \sum_{m^{'}}\sum_{n^{'}} a_{m^{'}n^{'}} \left\langle  \textbf{h}_{m^{'}}, \textbf{h}_{4 \lambda + l} \right\rangle  \left\langle  \textbf{h}_{\lambda + n^{'}}, \textbf{h}_{4 \lambda + k} \right\rangle \\
&= \textbf{E}_{i}^{T} \textbf{E}_{j}\left( l,k\right),
\end{align*}
if $a_{m^{'}n^{'}} = a_{mn}$. Let $a_{m^{'}n^{'}} = a_{t^{'}}$. $t^{'}$ is given by, $t^{'} = f^{-1}\left( f\left( m\right) \oplus f\left( k\right) \oplus f\left( l\right) \oplus f\left( n\right) \oplus f\left( k\right) \oplus f\left( l\right) \right) = f^{-1}\left( f\left( m\right) \oplus f\left( n\right) \right) = t$. 
Therefore, we can see that $\textbf{E}_{i}^{T} \textbf{E}_{j}$ is symmetric. Using the above arguments, it is also easly seen that the diagonal elements of the matrix $\textbf{E}_{i}^{T} \textbf{E}_{j}$ are identical. 

Hence, we have shown that the matrix $ \textbf{R}_{2}$ is block diagonal with $4$ blocks, each of size $ \lambda \times \lambda$.

\end{proof}

\section{Structure of $ \textbf{R} $ matrix obtained from construction II}
\label{app_r_mat_struct_const_2}
The STBC $ \textbf{X}$ can be written as 
\begin{equation*}
\textbf{X} = \sum_{i=1}^{K} x_{i} \textbf{A}_{i},
\end{equation*}
where $x_{i} = x_{iI} + jx_{iQ}$. Tweaking the system model in section \ref{sec2}, we can get a generator matrix for this STBC as
\begin{equation*}
\textbf{G}^{'} = \left[vec\left(\textbf{A}_{1}\right) ~ vec\left(\textbf{A}_{2}\right) ~ \cdots ~ vec\left(\textbf{A}_{K}\right) ~ \right].
\end{equation*}
Hence, \eqref{system_model} can be written as
\begin{equation*}
vec\left(\textbf{Y}\right) = \textbf{H}_{eq}^{'}\tilde{\textbf{x}} + vec\left(\textbf{N}\right),
\end{equation*}
where $\textbf{H}_{eq}^{'} \in \mathbb{C}^{n_{r}n_{t} \times K}$ is given by
$\textbf{H}_{eq}^{'} = \left(\textbf{I}_{n_{t}} \otimes \textbf{H}\right) \textbf{G}^{'},$ 
 and 
$\tilde{\textbf{x}} = \left[x_{1}, x_{2} . . . , x_{K}\right],$ 
with each $x_{i}$ drawn from a 2-dimensional constellation.
It can be easily seen that $\textbf{H}_{eq} = \check{\textbf{H}_{eq}^{'}}$. 

Let the QR decomposition of the complex matrix $\textbf{H}_{eq}^{'}$ yield matrices $\textbf{Q}^{'}$ and $ \textbf{R}^{'}$. Using the relation: If $ \textbf{A} = \textbf{B} \textbf{C}$, then $ \check{\textbf{A}} = \check{\textbf{B}} \check{\textbf{C}}$, we can see that $ \textbf{R} = \check{\textbf{R}^{'}}$. 
The QR decomposition of a complex matrix yields a unitary $ \textbf{Q}$ matrix and an upper triangular matrix $ \textbf{R}$ with real diagonal entries. Hence, the diagonal entries of the matrix $ \textbf{R}^{'}$ are real. Since $ \textbf{R} = \check{\textbf{R}^{'}}$, we'll have $ \textbf{R}\left( 2i-1 , 2i\right) = 0$ for $ i = 1,... K$. Hence, the STBC $ \textbf{X}$ exhibits a block orthogonal property with parameters $ \left( K, 2, 1\right) $.

\section{Structure of $ \textbf{R} $ matrix obtained from construction III}
\label{app_r_mat_struct_const_3}
Let the $\textbf{R}$ matrix for this code have the following structure:
\begin{equation*}
\textbf{R} = \left[\begin{array}{cc}
\textbf{R}_{1} & \textbf{E}\\
\textbf{0} & \textbf{R}_{2}\\
\end{array}\right],
\end{equation*}
where $ \textbf{R}_{1}$, $ \textbf{E}$ and $ \textbf{R}_{2}$ are $ 2K \times 2K$ matrices. 

From \cite{JiR}, it can be easily seen that $ \textbf{R}_{1}$ has a block diagonal structure with two blocks, and each block of the size $ K \times K$. 
\begin{equation*}
\textbf{R}_{1} = \left[\begin{array}{cc}
\textbf{R}_{11} & \textbf{0}\\
\textbf{0} & \textbf{R}_{12}\\
\end{array}\right],
\end{equation*}
where $ \textbf{R}_{11}$ and $ \textbf{R}_{12}$ are $ K \times K$ upper triangular matrices. 
\begin{proposition}
\label{r1_struct_prop_const_3}
The non-zero blocks of the matrix $ \textbf{R}_{1}$ are equal i.e., $ \textbf{R}_{11} = \textbf{R}_{12}$. 
\end{proposition}
\begin{proof}
Proof is similar to the proof of Proposition \ref{r1_struct_prop_const_1}.
\end{proof}

The structure of the matrix $ \textbf{E}$ is described in the following proposition.

\begin{proposition}
\label{e_struct_prop_const_3}
The matrix $ \textbf{E}$ is of the form
\begin{equation*}
\textbf{E} = \left[\begin{array}{cc}
\textbf{E}_{1} & -\textbf{E}_{2} \\
\textbf{E}_{2} & \textbf{E}_{1} \\
\end{array}\right],
\end{equation*}
where $ \textbf{E}_{i}$, $i = 1,...,4$ are $ K \times K$ matrices.
\end{proposition} 
\begin{proof}
Proof is similar to the proof of Proposition \ref{e_struct_prop_const_1}.
\end{proof}

\begin{proposition}
\label{r2_struct_prop_const_3}
The matrix $ \textbf{R}_{2}$ is block diagonal with $2$ blocks, each of size $ K \times K$.
\end{proposition}
\begin{proof}
Proof is similar to the proof of Proposition \ref{r2_struct_prop_const_1}.
\end{proof}

\section{Structure of $ \textbf{R} $ matrix obtained from construction IV}
\label{app_r_mat_struct_const_4}
As only rate-1 CIODs are considered in this construction, this can only be done for either $ 2 \times 2$ CIODs or $ 4 \times 4$ CIODs. The structure of the $ \textbf{R}$ matrix obtained from the $2 \times 2$ CIOD is the same as the structure of $ \textbf{R}$ matrix obtained from the construction III. The proof of the structure is also the same as given in Appendix \ref{app_r_mat_struct_const_4}. We now consider the structure of the $ \textbf{R}$ matrix obtained from using a $ 4 \times 4$ CIOD. 
Let the $\textbf{R}$ matrix for this code have the following structure:
\begin{equation*}
\textbf{R} = \left[\begin{array}{cc}
\textbf{R}_{1} & \textbf{E}\\
\textbf{0} & \textbf{R}_{2}\\
\end{array}\right],
\end{equation*}
where $ \textbf{R}_{1}$, $ \textbf{E}$ and $ \textbf{R}_{2}$ are $ 8 \times 8$ matrices. 

From \cite{JiR}, it can be easily seen that $ \textbf{R}_{1}$ has a block diagonal structure with $4$ blocks, and each block of the size $ 2 \times 2$. 
\begin{equation*}
\textbf{R}_{1} = \left[\begin{array}{cccc}
\textbf{R}_{11} & \textbf{0} & \textbf{0} & \textbf{0}\\
\textbf{0} & \textbf{R}_{12} & \textbf{0} & \textbf{0}\\
\textbf{0} & \textbf{0} & \textbf{R}_{13} & \textbf{0}\\
\textbf{0} & \textbf{0} & \textbf{0} & \textbf{R}_{14}\\
\end{array}\right],
\end{equation*}
where $ \textbf{R}_{1i}$ are $ 2 \times 2$ upper triangular matrices for $i= 1,...,4$. 
\begin{proposition}
\label{r1_struct_prop_const_4}
The non-zero blocks of the matrix $ \textbf{R}_{1}$ are such that $ \textbf{R}_{11} = \textbf{R}_{12}$ and $ \textbf{R}_{13} = \textbf{R}_{14}$. 
\end{proposition}
\begin{proof}
Proof is similar to the proof of Proposition \ref{r1_struct_prop_const_1}.
\end{proof}

The structure of the matrix $ \textbf{E}$ is described in the following proposition.

\begin{proposition}
\label{e_struct_prop_const_4}
The matrix $ \textbf{E}$ is of the form
\begin{equation*}
\textbf{E} = \left[\begin{array}{cccc}
\textbf{E}_{1} & -\textbf{E}_{2} & \textbf{E}_{5} & -\textbf{E}_{6}\\
\textbf{E}_{2} & \textbf{E}_{1} & \textbf{E}_{6} & \textbf{E}_{5}\\
\textbf{E}_{3} & -\textbf{E}_{4} & \textbf{E}_{7} & -\textbf{E}_{8}\\
\textbf{E}_{4} & \textbf{E}_{3} & \textbf{E}_{8} & \textbf{E}_{7}\\
\end{array}\right],
\end{equation*}
where $ \textbf{E}_{i}$, $i = 1,...,8$ are $ 2 \times 2$ matrices.
\end{proposition} 
\begin{proof}
Proof is similar to the proof of Proposition \ref{e_struct_prop_const_1}.
\end{proof}

\begin{proposition}
\label{r2_struct_prop_const_4}
The matrix $ \textbf{R}_{2}$ is block diagonal with $2$ blocks, each of size $ 2 \times 2$.
\end{proposition}
\begin{proof}
Proof is similar to the proof of Proposition \ref{r2_struct_prop_const_1}.
\end{proof}


\begin{thebibliography}{1}

\bibitem{HaH}
B. Hassibi and B. Hochwald, ``High-rate codes that are linear in space and time,`` IEEE Trans. Inf. Theory, vol. 48, no. 7, pp. 1804-1824, July 2002.

\bibitem{ViB}
E. Viterbo and J. Boutros, ``A Universal Lattice Code Decoder for Fading Channels,`` IEEE Trans. Inf. Theory, vol. 45, no. 5, pp. 1639-1642, July 1999.

\bibitem{RGYZ}
T. P. Ren, Y. L. Guan, C. Yuen and E. Y. Zhang, ``Block-Orthogonal Space–Time Code Structure and Its Impact on QRDM Decoding Complexity Reduction,`` IEEE journal of Selected Topics in Signal Processing, vol. 5, issue 8, pp. 1438-1450, Nov. 2011.

\bibitem{TJC}
V. Tarokh, H. Jafarkhani and A. R. Calderbank, ``Space-Time Block Codes from Orthogonal Designs,`` IEEE Trans. Inf. Theory, vol. 45, no. 5, pp. 1456-1467, July 1999.

\bibitem{Li}
X.B. Liang, ``Orthogonal Designs with Maximal Rates,`` IEEE Trans. Inf. Theory, vol.49, no. 10, pp. 2468-2503, Oct. 2003. 

\bibitem{TiH}
O. Tirkkonen and A. Hottinen, ``Square-Matrix Embeddable Space-Time Block Codes for Complex Signal Constellations,`` IEEE Trans. Inf. Theory, vol. 48, no. 2, pp. 384-395, Feb. 2002.

\bibitem{DYT}
D. N. Dao, C. Yuen, C. Tellambura, Y. L. Guan and T. T. Tjhung, ``Four-Group Decodable Space-Time Block Codes,`` IEEE Trans. Signal Processing, vol. 56, no. 1, pp. 424-430, Jan. 2008.

\bibitem{KaR}
S. Karmakar and B. S. Rajan, ``Multigroup Decodable STBCs From Clifford Algebras,`` IEEE Trans. Inf. Theory, vol. 55, no. 1, pp. 223-231, Jan. 2009.

\bibitem{KaR1}
S. Karmakar and B. S. Rajan, ``High-rate, Multi-Symbol-Decodable STBCs from Clifford Algebras,`` IEEE Transactions on Inf. Theory, vol. 55, no. 06, pp. 2682-2695, Jun. 2009.

\bibitem{BHV}
E. Biglieri, Y. Hong and E. Viterbo, ``On Fast-Decodable Space-Time Block Codes,`` IEEE Trans. Inf. Theory, vol. 55, no. 2, pp. 524-530, Feb. 2009.

\bibitem{VHO}
R. Vehkalahti, C. Hollanti and F. Oggier, ``Fast-Decodable Asymmetric Space-Time Codes from Division Algebras,`` available online at arXiv, arXiv:1010.5644v1 [cs.IT].

\bibitem{RGYS}
T. P. Ren, Y. L. Guan, C. Yuen and R. J. Shen, ``Fast-Group-Decodable Space-Time Block Code,`` Proceedings IEEE Information Theory Workshop, (ITW 2010), Cairo, Egypt, Jan. 6-8, 2010, available online at http://www1.i2r.a-star.edu.sg/cyuen/publications.html.

\bibitem{SiB}
M. O. Sinnokrot and J. Barry, ``Fast Maximum-Likelihood Decoding of the Golden Code,`` IEEE Transactions on Wireless Commun., vol. 9, no. 1, pp. 26-31, Jan. 2010.

\bibitem{BRV}
J. C. Belfiore, G. Rekaya and E. Viterbo, ``The golden code: A 2x2 full-rate space–time code with non-vanishing determinants,`` 
IEEE Trans. Inf. Theory, vol. 51, no. 4, pp. 1432-1436, Apr. 2005.

\bibitem{KaC}
S. Kahraman and M. E. Celebi, ``Dimensionality Reduction for the Golden Code with Worst-case Complexity of $O\left( m^2\right) $,`` available online at http://istanbultek.academia.edu/SinanKahraman .

\bibitem{RaR}
G. S. Rajan and B. S. Rajan, ``Multi-group ML Decodable Collocated and Distributed Space Time Block Codes,`` IEEE Trans. Inf. Theory, vol. 56, no. 7, pp. 3221-3247, July 2010.

\bibitem{KhR}
Z. Ali Khan Md., and B. S. Rajan, ``Single Symbol Maximum Likelihood Decodable Linear STBCs,`` IEEE Trans. Inf. Theory, vol. 52, no. 5, pp. 2062-2091, May 2006.

\bibitem{SeRS}
B. A. Sethuraman, B. S. Rajan and V. Shashidhar, ``Full-diversity, high-rate space-time block codes from division algebras,``
IEEE Trans. Inform. Theory, vol. 49, pp. 2596-2616, Oct 2003.

\bibitem{ShRS}
V. Shashidhar, B. S. Rajan and B. A. Sethuraman, ``Information-lossless space-time block codes from crossed-product algebras,``
IEEE Trans. Inf. Theory, vol. 52, no. 9, pp. 3913–3935, Sep 2006.

\bibitem{DCB}
O. Damen, A. Chkeif, and J.C. Belfiore, ``Lattice Code Decoder for Space-Time Codes,`` IEEE Communication Letters, vol. 4, no. 5, pp. 161-163, May 2000.

\bibitem{JiR}
G. R. Jithamithra and B. S. Rajan, ``Minimizing the Complexity of Fast Sphere Decoding of STBCs,`` available online at arXiv, arXiv:1004.2844v2 [cs.IT], 22 May 2011.

\bibitem{HLRVV}
C. Hollanti, J. Lahtonen, K. Ranto, R. Vehkalahti and E. Viterbo, ``On the algebraic structure of the Silver code: A 2 × 2
Perfect space-time code with non-vanishing determinant,`` 
in Proc. of IEEE Inf. Theory Workshop, Porto, Portugal, May 2008.

\bibitem{PaR}
K. P. Srinath and B. S. Rajan, ``Low ML-Decoding Complexity, Large Coding Gain, Full-Rate, Full-Diversity STBCs for 2x2 and 4x2 MIMO Systems,`` IEEE Journal of Selected Topics in Signal Processing: Special issue on Managing Complexity in Multiuser MIMO Systems, vol. 3, no. 6, pp. 916-927, Dec. 2009.

\bibitem{CLRS}
T. H. Cormen, C. E. Leiserson, R. L. Rivest, C. Stein, ``Introduction to algorithms,``
Third edition, MIT Press, Sep 2009.

\end{thebibliography}
\end{document}